\documentclass{article}

\usepackage{amssymb}

\usepackage[textsize=footnotesize]{todonotes} 
\usepackage[blocks]{authblk}
\usepackage[utf8]{inputenc}
\usepackage{aliascnt}
\usepackage{amsmath}
\usepackage{subcaption}
\usepackage{graphicx}
\usepackage{amsthm}
\usepackage{mathrsfs}
\usepackage{aliascnt}
\usepackage{enumerate}
\usepackage[shortlabels]{enumitem}
\usepackage{bbm}
\usepackage{bm}
\usepackage{color}
\usepackage{xargs}
\usepackage{blkarray}
\usepackage{enumitem}
\usepackage{stmaryrd}
\usepackage{natbib}
\usepackage{float}
\usepackage{dsfont}
\usepackage[ruled,linesnumbered]{algorithm2e}
\usepackage{multirow}
\usepackage[colorlinks,citecolor=red,linkcolor=blue]{hyperref}
\usepackage[nameinlink,capitalize]{cleveref}
\usepackage{geometry} 
\makeatletter
\def\algbackskip{\hskip-\ALG@thistlm}
\makeatother

\makeatletter
\newcommand*{\addFileDependency}[1]{
  \typeout{(#1)}
  \@addtofilelist{#1}
  \IfFileExists{#1}{}{\typeout{No file #1.}}
}
\makeatother

\newfont{\bg}{cmr10 scaled\magstep2}

\newcommandx{\norm}[2][2=]{\| #1 \|_{#2}}
\newcommandx{\nint}[2]{\left\{ #1, \ldots, #2 \right\}}
\newcommandx{\maxi}[2][2=]{\underset{#2}{\operatorname{max}}\left\{#1\right\}}
\newcommandx{\mini}[2][2=]{\underset{#2}{\operatorname{min}}\left\{#1\right\}}
\newcommandx{\argmaxi}[2][2=]{\underset{#2}{\operatorname{argmax}}\left\{#1\right\}}
\newcommandx{\argmini}[2][2=]{\underset{#2}{\operatorname{argmin}}\left\{#1\right\}}
\newcommand{\xupdownarrow}[1]{%
  {\left\updownarrow\vbox to #1{}\right.\kern-\nulldelimiterspace}
}

\newcommandx{\rk}[1]{\operatorname{rk}(#1)}
\newcommandx{\prob}[1]{\mathbb{P}\left\{#1 \right\} }
\newcommandx{\e}{\mathrm{e}}

\newtheorem{theorem}{Theorem}
\newtheorem{assumption}{\textbf{H}}
\Crefname{assumption}{\textbf{H}}{\textbf{H}}
\crefname{assumption}{\textbf{H}}{\textbf{H}}
\newtheorem{lemma}{Lemma}
\Crefname{lemma}{Lemma}{Lemma}
\crefname{lemma}{lemma}{lemma}

\Crefname{corollary}{Corollary}{Corollary}
\crefname{corollary}{corollary}{corollary}

\Crefname{remark}{Remark}{Remark}
\crefname{remark}{remark}{remark}

\newtheorem{proposition}{Proposition}
\Crefname{proposition}{Proposition}{Proposition}
\crefname{proposition}{proposition}{proposition}

\Crefname{definition}{Definition}{Definition}
\crefname{definition}{definition}{definition}

\Crefname{example}{Example}{Example}
\crefname{example}{example}{example}

\newcommand{\intervi}[1]{[#1]}
\def\eqsp{\,}
\newcommand{\ensemble}[2]{\left\{#1\,:\eqsp #2\right\}}
\newcommand{\set}[2]{\ensemble{#1}{#2}}
\def\rank{\operatorname{rank}}

\def\nset{\mathbb{N}}

\begin{document}

\title{Low-rank model with covariates for count data with missing values}

  \author[1,2]{Genevi\`{e}ve Robin}
  \author[1,2]{Julie Josse}
\author[1,2]{\'{E}ric Moulines}  
\author[3]{Sylvain Sardy}
\affil[1]{CMAP, UMR 7641, \'{E}cole Polytechnique }
\affil[2]{XPOP, INRIA }
\affil[3]{Department of Mathematics, Universit\'{e} de Gen\`{e}ve}
\maketitle

\begin{abstract}
Count data are collected in many scientific and engineering tasks including image processing, single-cell RNA sequencing  and ecological studies. Such data sets often contain missing values, for example because some ecological sites cannot be reached in a certain year. In addition, in many instances, side information is also available, for example covariates about ecological sites or species. Low-rank methods are popular to denoise and impute count data, and benefit from a substantial theoretical background. Extensions accounting for covariates have been proposed, but to the best of our knowledge their theoretical and empirical properties have not been thoroughly studied, and few softwares are available for practitioners.\\
We propose a complete methodology called LORI (Low-Rank Interaction), including a Poisson model, an algorithm, and automatic selection of the regularization parameter, to analyze count tables with covariates. We also derive an upper bound on the estimation error.
We provide a simulation study with synthetic data, revealing empirically that LORI improves on state of the art methods in terms of estimation and imputation of the missing values. We illustrate how the method can be interpreted through visual displays with the analysis of a well-know plant abundance data set, and show that the LORI outputs are consistent with known results. Finally we demonstrate the relevance of the methodology by analyzing a water-birds abundance table from the French national agency for wildlife and hunting management (ONCFS). The method is available in the R package \texttt{lori} on the Comprehensive Archive Network (CRAN).\end{abstract}

\noindent%
\vfill


\section{Introduction}
\label{introduction}
Let $Y$ be an $n\times p$ observation matrix of counts, $R\in\mathbb{R}^{n\times K_1}$ and $C\in\mathbb{R}^{p\times K_2}$ be two matrices containing row and column covariates, respectively. In our ecological application in \Cref{eco-data}, rows of the contingency table represent ecological sites, and columns represent years. For $(i,j)\in \{1,\ldots,n\}\times\{1,\ldots,p\}$, $Y_{ij}$ counts the abundance of water-birds measured in site $i$ during the year $j$. The row feature $R_{i\ell}$, $\ell \in \{1,\ldots, K_1\}$ embeds geographical information about the site $i$ (latitude, longitude, distance to coast, etc.) while the column feature $C_{j \ell}$, $\ell \in \{1,\ldots, K_2\}$ codes meteorological characteristics of the year $j$ (precipitation, etc.). In addition, some entries of $Y$ are missing. For example ecological sites are sometimes inaccessible because of meteorological or political conditions, and therefore cannot be counted.\\

Such count tables are often analyzed using \textit{low-rank models} \citep{Greenacre84,Good85, falg98,Christensen90,Gow11c, JosseFithian2016}, imposing a low-rank structure to an underlying parameter matrix. We assume a probabilistic framework with independent entries $Y_{ij}$ following a Poisson model
\begin{equation}
\label{eq:poisson}
Y_{ij}\sim \mathcal{P}(\mathrm{e}^{X^{*}_{ij}}), (i,j)\in \intervi{n}\times\intervi{p},
\end{equation}
and focus on the estimation of $X^{*}$ based on a low-rank assumption.
The \textit{generalized additive main effects and multiplicative interaction} model, or \textit{row-column} model (see, e.g., \citet{Good85, falg98}), assuming
\begin{equation}
\label{eq:saturated-log-linear}
X^*_{ij}=\mu^*+\alpha^*_i+\beta^*_j+\Theta^*_{ij}\text{, \quad }\operatorname{rank}(\Theta^*)\leq \min(n-1, p-1),
\end{equation}
is adapted to do so. In this model, $\mu^*$ is an offset, the terms which only depend on the index of the row or column ($\alpha^*_i$ and $\beta^*_j$) are called \textit{main effects}, and the terms which depend on both (here $\Theta^*_{ij}$) are called \textit{interactions} \citep[Section 4.1.2, p.87]{Kateri2014}. \\

A natural idea to incorporate covariates in this framework, is to express the row and column effects $\alpha^*_i$ and $\beta^*_j$ as regression terms on the covariates. In other words, for $\mu^{*}\in\mathbb{R}$, $\alpha^{*} \in \mathbb{R}^{K_1}$, $\beta^{*}\in\mathbb{R}^{K_2}$ and $\Theta^{*}\in \mathbb{R}^{n\times p}$,
\begin{equation}
\label{eq:x-model2}
X^{*}_{ij} = \mu^{*} + \underbrace{\sum_{k=1}^{K_1}R_{ik}\alpha^*_k}_{\text{row effect}} + \underbrace{\sum_{l=1}^{K_2}C_{il}\beta_l^*}_{\text{column effect}}  + \Theta^{*}_{ij}, \quad \operatorname{rank}(\Theta^*)\leq \min(n - 1, p-1).
\end{equation}
Such an extension is useful in practice for two main reasons. First, estimated covariates coefficients (and in particular their signs) can be used to determine whether the studied covariates have positive or negative effects on the counts; this is particularly useful in ecology, to check whether meteorological, geographical or political conditions favor or endanger species. Second, when the proportion of missing values is large, which is often the case in bird monitoring, incorporating (relevant) covariates can improve the imputation significantly.\\

Models related to \eqref{eq:x-model2} have been considered for statistical ecology applications in \citep{MEE3:MEE312163,10.7717/peerj.2885}. However, to the best of our knowledge, their theoretical and empirical properties have not been thoroughly studied. On the other hand, the literature on convex low-rank matrix estimation is abundant and benefits from a substantial theoretical background, but few software with ready to use solution are available for practitioners, and applications for count data outside image analysis \citep{Luisier2010a,Salmon14,Cao2016} and recommendation systems \citep{Gopalan14bayesiannonparametric} have not been attempted. The scope of this paper is to develop a complete methodology for the inference of model \eqref{eq:x-model2}, bridging the gap between convex low-rank matrix completion and model-based count data analysis.\\

After detailing related work in \Cref{related-work}, we introduce in \Cref{model} a general model which includes \eqref{eq:x-model2}; we propose an estimation procedure through the minimization of a data fitting term
penalized by the nuclear norm of the matrix $\Theta$, which acts as a convex relaxation of the rank constraint. Building up on existing results on nuclear norm regularized loss functions, we derive statistical guarantees in \Cref{guarantee}. In particular, we provide an upper bound for the Frobenius norm of the estimation error.
 In \Cref{lambda-choice}, we propose an optimization algorithm, and two methods to choose the regularization parameter automatically. We provide a simulation study in \Cref{experiments} revealing that LORI outperforms state-of-the-art methods when the proportion of missing values is large and the interactions are of significant order compared to the main effects. 
In \Cref{real-data}, we show on plant abundance data with side information, how the results of our procedure can be interpreted through visual displays. In particular, the arising interpretation is consistent with known results from the original study \citep{aravo}. 
In \Cref{eco-data}, we use LORI to analyze a water-birds abundance data set from the French national agency for wildlife and hunting management (ONCFS). The proofs of the statistical guarantees are postponed to the appendix, and the method is available as an R package \citep{R} called \texttt{lori}  (LOw-Rank Interaction) on the Comprehensive R Archive Network (CRAN) at \url{https://CRAN.R-project.org/package=lori}. 

\subsection{Related work}
\label{related-work}
Model \eqref{eq:x-model2} is closely related to other models previously suggested in the statistical ecology literature to analyze count tables with row and column covariates. For instance,
\citet{MEE3:MEE312163} and  \citet{10.7717/peerj.2885} suggested the following model:
\begin{equation}
\label{eq:saturated-log-linear-cov}
X^*_{ij}=\mu^*+\alpha^*_i+\beta^*_j+\epsilon_{RC}R_iC_j,
\end{equation}
with $R_i$, $1\leq i\leq n$ a row trait and  $C_j$, $1\leq j\leq p$ a column trait. The interaction between covariates is modeled by $\epsilon_{RC}R_iC_j$, where $\epsilon_{RC}$ is an unknown parameter measuring the strength of the interaction between the two traits. The main difference with model \eqref{eq:x-model2} is that we incorporate the covariates in the main effects rather than the interactions, which leads to different interpretations. In terms of estimation properties, the main advantage of \eqref{eq:x-model2} is that, as long as $K_1\leq n$ and $K_2\leq p$, we estimate less parameters. This is an important point for us since in many applications we are interested in (see e.g. \Cref{eco-data}), a large proportion of entries is missing, limiting the amount of available data. Finally, model \eqref{eq:saturated-log-linear-cov} was developed with the aim of testing significant associations between covariates, and its theoretical and empirical estimation properties, as far as we know, were not studied. \\

In the low-rank matrix completion literature, related approaches for count matrix recovery and dimensionality reduction can be embedded within the framework of low-rank exponential family estimation \citep{Collins01ageneralization, deLeeuw2006PCA, li2013simple,  JMLR:v17:14-534, Edgar16} as well as its Bayesian counterpart \citep{Ghahramani2009, Gopalan}. In terms of statistical guarantees, the theoretical performance of nuclear norm penalized estimators for Poisson denoising has been studied in \citet{Cao2016}, where the authors prove uniform bounds on the empirical error risk. Estimation rates are also given in \citet{Lafond2015}, where optimal bounds are proved for matrix completion in the exponential family. These two papers do not account for available covariates. \\

More recently, \citet{2018_AOAS_cmr} developped a probabilistic PCA framework for the exponential family, where covariates can be included in the parameter space. \citet{fithian2013scalable} present a variety of low-rank problems including the generalized nuclear norm penalty \citep{Angst:2011:GTA:2355573.2356388}, that can be used to include row and column covariates.  Similar estimation problems were also considered, e.g., in \citet{Agarwal:2009:RLF:1557019.1557029, Abernethy:2009:NAC:1577069.1577098}. However, to the best of our knowledge, these papers did not provide statistical guarantees and the practical advantages of such extensions compared to classical low-rank methods have not been thoroughly studied.

\section{General model and estimation}
\label{model}
We now introduce a general version of the model described in the previous section.  First, we relax the Poisson model and replace it with the following assumption on the distribution of $Y_{ij}$, $(i,j)\in\intervi{n}\times\intervi{p}$.
\begin{assumption}
\label{ass:boundedness}
The random variables $Y =\{Y_{i,j} \}_{(i,j) \in \intervi{n} \times \intervi{p}}$ are independent and there exist $\gamma > 0$, $\sigma_{-} > 0$ and $\sigma_{+} < \infty$ such that for all $i \in \intervi{n}$ and $j \in \intervi{p}$ $$\mathrm{e}^{-\gamma}\leq \mathbb{E}[Y_{ij}]\leq \mathrm{e}^{\gamma} \text{ and }\sigma_{-}^2\leq\operatorname{var}[Y_{ij}]\leq \sigma_{+}^2.$$
\end{assumption}
We define $X^*_{ij}$ for all $(i,j)\in\intervi{n}\times\intervi{p}$ by
\begin{equation}
\label{eq:x-def}
X^*_{ij} = \operatorname{argmin}_{x\in\mathbb{R}} \{ -\mathbb{E}[Y_{ij}]x + \exp(x) \} \eqsp.
 \end{equation}
In other words, we do not assume that the random variable $Y_{ij}$ follows a Poisson distribution. The  target parameter  $X_{ij}^*$ minimizes the Kullback-Leibler divergence between the distribution of $Y_{ij}$ and a Poisson distribution.
Let us also generalize the decomposition introduced in \eqref{eq:x-model2}. Let $S_1$ and $S_2$ be fixed linear subspaces of $\mathbb{R}^n$ and $\mathbb{R}^p$ respectively. Let $P_1$ and $P_2$ be the orthogonal projection matrices on $S_1$ and $S_2$, $\mathcal{P}^{\perp}: X\in\mathbb{R}^{n\times p} \mapsto P_1XP_2^{\top}$, $\mathcal{P}:X\in\mathbb{R}^{n\times p} \mapsto X - \mathcal{P}^{\perp}(X)$, $\mathcal{X}_0 \subset \{X\in\mathbb{R}^{n\times p};\mathcal{P}^{\perp}(X) =0 \}$ and $\mathcal{T} = \{X\in\mathbb{R}^{n\times p};\mathcal{P}(X) =0 \}$. We denote 
\begin{equation}
\label{eq:rank-r}
r= \max\left( \set{\rank(A)}{A \in \mathcal{X}_0}\right).
\end{equation}
Consider the following decomposition:
\begin{equation}
\label{eq:general-model}
X^* = X_0^* + \Theta^*, \quad X_0^*\in \mathcal{X}_0, \Theta^* \in \mathcal{T}.
\end{equation}
Denote, for $m\geq 1$, $\mathbbm{1}_m$ the vector of ones of length $m$. Model \eqref{eq:x-model2} is included in \eqref{eq:general-model} by setting $S_1 = \{u\in\mathbb{R}^n; \mathbbm{1}_n^{\top}u = 0\}$, $S_2 = \{v\in\mathbb{R}^p; \mathbbm{1}_p^{\top}v = 0\}$, and
\[\mathcal{X}_0 = \left\{ \left(\mu+\sum_{k=1}^{K_1} R_{ik}\alpha_k + \sum_{k=2}^{K_2} C_{ik}\beta_k\right)_{(i,j)\in\intervi{n}\times\intervi{p}}; \mu\in \mathbb{R}, \alpha\in\mathbb{R}^{K_1}, \beta\in\mathbb{R}^{K_2} \right\}.
\]
The dimension of this subspace is at most $1 + K_1+K_2$ and the rank of a matrix in $\mathcal{X}_0$ is less that $3$. We finally consider a setting with missing observations. Denote by $\Omega \subset \intervi{n} \times \intervi{p}$ the set of observed entries:  $(i,j)\in\Omega$ if and only if $Y_{ij}$ is observed. Define also the random variables $(\omega_{ij})$ such that $\omega_{ij} = 1$ if $Y_{ij}$ is observed and $\omega_{ij} = 0$ otherwise. 
We assume that $(\omega_{ij})$ and $Y$ are independent, and a Missing Completely At Random (MCAR) scenario \citep{Little02} where $(\omega_{ij})$ are independent Bernoulli random variables. For $(i,j)\in\{1,\ldots,n\}\times\{1,\ldots,p\}$, we denote $\pi_{ij} = \mathbb{P}(\omega_{ij}=1)$. We assume the probability of observing any entry is positive, i.e. there exists $\pi>0$ such that
 \begin{equation}
\label{eq:mcar}
\min \set{\pi_{ij}}{(i,j) \in \intervi{n}\times\intervi{p}} =  \pi > 0 \eqsp.
\end{equation}
For $j\in\intervi{p}$, denote by $\pi_{.j} = \sum_{i=1}^{n}\pi_{ij}$ the probability of observing an element in the $j$-th column. Similarly, for $i\in\intervi{n}$, denote by $\pi_{i.} = \sum_{j=1}^{p}\pi_{ij}$ the probability of observing an element in the $i$-th row. We define the following upper bound:
\begin{equation}
\label{eq:beta}
\max \left(\set{\pi_{i.}}{i \in \intervi{n}} \cup \set{\pi_{.j}}{j \in \intervi{p}} \right)\leq \beta \eqsp.
\end{equation}
We can now define our data-fitting term:
\begin{equation}
\label{eq:likelihood}
\mathcal{L}(X) = \sum_{(i,j)\in\intervi{n}\times\intervi{p}}\omega_{ij}\left\{-Y_{ij}{X}_{ij} + \exp({X}_{ij})\right\}.
\end{equation}
Denote $\norm{\cdot}$ the operator norm (the largest singular value), $\norm{\cdot}[\infty]$ the infinity norm (the largest entry in absolute value) and $\norm{\cdot}[*]$ the nuclear norm (the sum of singular values). Our estimator of model \eqref{eq:x-model2}, for a given regularization parameter $\lambda>0$, is the minimizer of the data-fitting term \eqref{eq:likelihood} penalized by the nuclear norm of $\Theta$: 
\begin{equation}
\label{eq:estimator}
\begin{aligned}
& (\hat X_0,\hat\Theta)&& \in{\text{argmin}}\quad\mathcal{L}(X_0 +\Theta)+\lambda\norm{\Theta}[*],\\
&\text{such that} &&\norm{{X_0} + \Theta}[\infty]\leq \gamma, (i,j)\in\intervi{n}\times\intervi{p}\\
& && X_0\in\mathcal{X}_0,\quad \Theta \in \mathcal{T}.
\end{aligned}
\end{equation}
Denote by $\nabla \mathcal{L}(X) = \sum_{(i,j)\in\intervi{n}\times\intervi{p}}\omega_{ij}\left\{-Y_{ij} + \exp({X}_{ij})\right\}E_{ij}$, where $(E_{ij})$ are the matrices of the canonical basis of $\mathbb{R}^{n\times p}$, the gradient of $\mathcal{L}$ at $X$.
Denote also $\partial^2\mathcal{L}/\partial x_{ij}^2$ the second derivative of $\mathcal{L}$ with respect to the $(i,j)$-th coordinate. Consider the following condition:
\begin{assumption}
\label{ass:L}
The function $\mathcal{L}$ is strongly convex and smooth on $[-\gamma - \varepsilon, \gamma +\varepsilon]^{n\times p}$ for some $\varepsilon>0$.
There exist $\sigma_->0$ and $\sigma_+<\infty$ such that for all $X\in [-\gamma - \varepsilon, \gamma+ \varepsilon]^{n\times p}$ and $(i,j)\in\intervi{n}\times\intervi{p}$, $\sigma_-^2\leq \partial^2\mathcal{L}(X)/\partial x_{ij}^2\leq \sigma_+^2$.
\end{assumption}

\subsection{Statistical guarantees} \label{guarantee}
We now derive an upper bound on the Frobenius estimation error of estimator~\eqref{eq:estimator}. Let $(\epsilon_{ij})$, $(i,j)\in\intervi{n}\times\intervi{p}$ be i.i.d. Rademacher random variables independent of $Y$ and $\Omega$ and define
\begin{equation}
\label{eq:def-SigmaR}
\Sigma_R = \sum_{(i,j)\in\intervi{n}\times\intervi{p}}\epsilon_{ij}\omega_{ij}E_{ij}.
\end{equation}
\begin{theorem}
\label{th:global-bound-1} 
Assume \textbf{H}~\ref{ass:boundedness}-\ref{ass:L}, and $\lambda \geq 2\norm{\nabla\mathcal{L}(X^*)}$. Then for all $n,p\geq 1$, with probability at least $1-8(n+p)^{-1}$,
\begin{equation}
\label{eq:global-bound-1}
\left\|X^*-\hat{X}\right\|_{F}^2\leq \frac{C}{\pi^2}\left(\left[\frac{\lambda^2}{\sigma_-^4} + (\mathbb{E}\norm{\Sigma_R})^2\gamma^2\right](\operatorname{rank}(\Theta^*)+r) + \log(n+p) \right),
\end{equation}
where $r$, $\gamma$ are defined in \eqref{eq:rank-r} and \eqref{eq:estimator}, $C$ is a numerical constant whose value can be found in the proof and which is independent of $n$, $p$ and $X^*$.
\end{theorem}
\begin{proof}
The proof is postponed to \Cref{proof:global-bound-1}.
\end{proof}
We then control $\mathbb{E}\norm{\Sigma_R}$, and compute a value of $\lambda$ such that the condition $\lambda\geq 2\norm{\nabla \mathcal{L}(X^*)}$ holds with high probability. We will need the following additional assumption on the distribution of the counts:
\begin{assumption}
\label{ass:orlicz}
There exists $\delta>0$ such that for all $(i,j) \in \intervi{n} \times \intervi{p}$, 
\[
\mathbb{E}\left[\exp(|Y_{ij}|/\delta)\right]<+\infty \eqsp.
\]
\end{assumption}
Define the following quantities, with $C^*$ a numerical constant defined in \Cref{lemma:SigmaR} and $r$, $\beta$ and $\gamma$ defined in \eqref{eq:rank-r}, \eqref{eq:beta} and \eqref{eq:estimator} respectively:
\begin{equation}
\label{eq:def-phi}
\begin{aligned}
&\Phi_1 &&= 48\sigma_+^2\beta\log(n+p),\\
&\Phi_2 &&= 36\delta^2(\mathrm{e}-1)^2\log^1\left(1+8\delta^2\frac{np}{\beta\sigma_-^2}\right)\log^2(n+p)\\
&\Phi_3 &&= 4{C^*}^2\max(\beta^2,\log\{\min(n,p)\}).
\end{aligned}
\end{equation}
\begin{theorem}
\label{th:global-bound} 
Assume \Cref{ass:boundedness}-\Cref{ass:orlicz} and set 
\[
\lambda = \max\left\{4\sigma_{+}\sqrt{3\beta\log(n+p)}, 12\delta(e-1)\log\left(1+8\delta^2\frac{np}{\beta\sigma_-^2}\right)\log(n+p)\right\}.
\]
Then with probability at least $1-10(n+p)^{-1}$,
\begin{equation}
\label{eq:global-bound}
\left\|X^*-\hat{X}\right\|_{F}^2\leq \frac{C}{\pi^2}\left\{\left(\max(\Phi_1,\Phi_2) + \Phi_3\right)\left(\operatorname{rank}(\Theta^*)+r\right) + \log(n+p)\right\},
\end{equation}
where $C$ is a numerical constant independent of $n$, $p$ and $X^*$. 
\end{theorem}
\begin{proof}
The proof is postponed to \Cref{proof:global-bound}.
\end{proof}
We recover an upper bound of order $\operatorname{rank}(\Theta^*)\beta/\pi^2$, which is classical in low-rank matrix estimation and completion \citep{Klopp2014, Lafond2015} and equal to $\operatorname{rank}(\Theta^*)\max(n,p)/\pi$ when the sampling is almost uniform ($c_1\pi\leq\pi_{ij}\leq c_2\pi$). The additional term $r\beta/\pi^2$ accounts for explicit modeling of the covariates in the main effects. The constant term appearing in bound \eqref{eq:global-bound} grows linearly with the upper bound $\sigma_{+}^2$ and quadratically with the inverse of $\sigma_{-}^2$. This means that by relaxing Assumption~\ref{ass:boundedness} to allow $\operatorname{var}(Y_{ij})$ to grow as fast as $\log(n+p)$ or decrease as fast as $1/\log(n+p)$, we only lose a log-polynomial factor in bound \eqref{eq:global-bound}.

\section{Algorithm and selection of $\lambda$}
\label{lambda-choice}

\subsection{Optimization algorithm}
\label{optim}
In this section, we propose an algorithm to solve the initial \eqref{eq:estimator} for the initial model
\begin{equation*}
X^{*}_{ij} = \mu^{*} + \sum_{k=1}^{K_1}R_{ik}\alpha^*_k + \sum_{l=1}^{K_2}C_{il}\beta_l^* + \Theta^{*}_{ij},
\end{equation*}
using \textit{alternating minimization} \citep{csiszar_information_1984}, which consists in updating $\mu$, $\alpha$, $\beta$ and $\Theta$ alternatively, each time along a descent direction. Note that, in the algorithm and the entire numerical section, we relax the constraint $|\mu + R_{i,.}\alpha + C_{j,.}\beta+\Theta_{ij}|\leq \gamma$. Indeed, this constraint is mainly required to obtain statistical guarantees, and we observed that in practice, for $\gamma$ sufficiently large, this constraint is never reached. Denote 
$$\mathcal{F}(\mu,\alpha,\beta,\Theta) = \mathcal{L}((\mu + R_{i,.}\alpha + C_{j,.}\beta)_{i,j} + \Theta),$$
and $\nabla_{\Theta}\mathcal{F}$ the gradient of $\mathcal{F}$ with respect to $\Theta$  defined by 
$(\nabla_{\Theta}\mathcal{F}(\mu,\alpha,\beta,\Theta))_{ij} = -Y_{ij} + \exp(\mu + R_{i,.}\alpha+C_{j,.}\beta+\Theta_{ij})$ if $\omega_{ij}=1$ and $(\nabla_{\Theta}\mathcal{F}(\mu,\alpha,\beta,\Theta))_{ij} =0$ otherwise.
At every iteration we solve three sub-problems. The sub-problem in $\mu$ can be solved in closed form at each iteration; the updates in $\alpha$ and $\beta$ can be done simultaneously by estimating a Poisson generalized linear model (which can be done using standard algorithms implemented in available libraries); the update in $\Theta$ is along the proximal gradient direction, with a step size tuned using backtracking line search. Denote by $\mathcal{D}_{\lambda}$ the soft-thresholding operator of singular values at level $\lambda$ \cite[Section 2]{cai2010singular}. The procedure is sketched in Algorithm~\ref{algo:altmin-mis}.
\begin{figure}[H]
\footnotesize
\centering
\begin{minipage}{\linewidth}
\begin{algorithm}[H]
\textbf{Initialize} $\mu^{[0]}$, $\alpha^{[0]}$, $\beta^{[0]}$, $\Theta^{[0]}$\\
\textbf{For $t = 0,1,\ldots, T-1$}
\begin{itemize}
\item $\mu^{[t+1]}\in \operatorname{argmin} \mathcal{F}(\mu,\alpha^{[t]},\beta^{[t]},\Theta^{[t]})$,
\item $(\alpha^{[t+1]},\beta^{[t+1]})\in \operatorname{argmin} \mathcal{F}(\mu^{[t]},\alpha,\beta,\Theta^{[t]}),$
\item $\tau = 1$,
\item $\Theta^{[t+1]} = \mathcal{D}_{\lambda}[\Theta^{[t]} - \tau\mathcal{P}\{\nabla_{\Theta}\mathcal{F}(\mu^{[t]},\alpha^{[t]},\beta^{[t]},\Theta^{[t]})\}]$
\item While $\mathcal{F}(\mu^{[t+1]},\alpha^{[t+1]},\beta^{[t+1]},\Theta^{[t+1]})+ \lambda\norm{\Theta^{[t+1]}}[*] > \mathcal{F}(\mu^{[t+1]},\alpha^{[t+1]},\beta^{[t+1]},\Theta^{[t]})+ \lambda\norm{\Theta^{[t]}}[*] $:
\begin{itemize}
\item $\tau = \tau/2$
\item $\Theta^{[t+1]} = \mathcal{D}_{\lambda}[\Theta^{[t]} - \tau\mathcal{P}\{\nabla_{\Theta}\mathcal{F}(\mu^{[t]},\alpha^{[t]},\beta^{[t]},\Theta^{[t]})\}]$.
\end{itemize}
\end{itemize}
\textbf{Output} $\mu^{[T]}$, $\alpha^{[T]}$, $\beta^{[T]}$, $\Theta^{[T]}$
\caption{Alternating minimization for problem \eqref{eq:estimator}}
\label{algo:altmin-mis}
 \end{algorithm}
\end{minipage}
\end{figure}
Note that if $K_1+K_2> |\Omega|$, with $|\Omega|$ denoting the cardinality of $\Omega$, the update in $\alpha$ and $\beta$ does not have a unique solution. However in our targeted applications, typically $K_1+K_2\ll |\Omega|$. In the package \texttt{lori}, we additionally implemented a warm-start strategy \citep{friedman2007}, which consists in solving \eqref{eq:estimator} for a large value of $\lambda$, then sequentially decreasing $\lambda$ and solving the new problem using the previous estimate as a starting point. Thus, our implementation solves \eqref{eq:estimator} for the entire regularization path at once. Note that, even though our theoretical guarantees require a MCAR mechanism, the estimation method still holds when entries are \textit{missing at random} (\citet{Little02}, Section 1.3). Its imputation properties are illustrated in an ecological application in \Cref{eco-data}.

\subsection{Automatic selection of $\lambda$}
A common way to select the regularization parameter is cross-validation, which consists in erasing a fraction of the observed cells in $Y$, estimating a complete parameter matrix $\hat{X}$ for a range of $\lambda$ values, and choosing the parameter $\lambda$ that minimizes the imputation error. This can be performed directly using LORI without modifying the code. Indeed,
let $(\tilde\omega_{ij})$ denote the weights in $\{0,1\}$ indicating which entries are observed after removing some of them for cross-validation, and denote 
$$\mathcal{\tilde F}(\mu,\alpha,\beta,\Theta) = \sum_{(i,j)\in\intervi{n}\times\intervi{p}}\tilde\omega_{ij}\{-Y_{ij}(\mu+R_{i,.}\alpha +C_{j,.}\beta +\Theta_{ij})+\exp(\mu+R_{i,.}\alpha +C_{j,.}\beta +\Theta_{ij})\}.$$
The optimization problem becomes
\begin{equation}
\label{eq:opt-mis}
\begin{aligned}
& (\hat\mu,\hat\alpha,\hat\beta,\hat\Theta)&& \in{\text{argmin}}\quad\tilde{\mathcal{F}}(\mu,\alpha,\beta,\Theta)+\lambda\norm{\Theta}[*],\\
& &&\quad \Theta\in\mathcal{T}
\end{aligned}
\end{equation}
which can be solved using the method described in \Cref{optim} (see Algorithm~\ref{algo:altmin-mis}). However, cross-validation is computationally costly. We suggest an alternative method to cross-validation, inspired by \citet{Dono94b} and the work of \citet{Sardy16} on \textit{quantile universal threshold}. In  \Cref{prop:null-lambda} below, we define the so-called \textit{null-thresholding statistic} of estimator \eqref{eq:estimator}, a function of the data $\lambda_0({Y})$ for which the estimated interaction matrix $\hat{\Theta}^{\lambda_0(Y)}$ is null, and the same estimate $\hat{\Theta}^{\lambda}=0$ is obtained for any $\lambda \geq \lambda_0({Y})$. 
\begin{theorem}[Null-thresholding statistic]
\label{prop:null-lambda}
The estimated interaction matrix $\hat{\Theta}^{\lambda}$ for a regularization parameter $\lambda$ is null if and only if $\lambda \geq \lambda_{0}(Y)$,
where $\lambda_0(Y)$ is the null-thresholding statistic
\begin{equation}
\label{eq:null-lambda}
\lambda_0({Y}) = \left\|\nabla\mathcal{L}(\hat X_0)\right\|,\quad {\rm where }\quad  \hat X_0\in{\operatorname{argmin}_{X\in \mathcal{X}_0}}\quad\mathcal{L}(X).
\end{equation}
\end{theorem}
\begin{proof}
The proof is postponed to \Cref{lambda-proof}.
\end{proof}
Here,  $\left\| . \right\|$ is the operator norm (the largest singular value). We propose a heuristic selection of $\lambda$ based on this null-thresholding statistic $\lambda_0(Y)$. To explain further the procedure, we first need to define the following test:
\begin{equation}
\label{test}
\mathbf{H}_0:\Theta^*= 0\text{\, \, against the alternative \, \,}\mathbf{H}_1:\Theta^*\neq 0
\end{equation}
which tests whether the parameter matrix $X^*$ can be explained only in terms of linear combinations of the measured covariates.
For a probability $\varepsilon\in(0,1)$, consider the upper $\varepsilon$-quantile $\lambda_{\varepsilon}$ of the null-thresholding statistics, namely that satisfies $\mathbb{P}_{\mathbf{H}_0}\left(\lambda_0(Y)>\lambda_{\varepsilon}\right)<\varepsilon$. The test which consists in comparing the statistics $\lambda_0(Y)$ to $\lambda_{\varepsilon}$ is of level $1-\varepsilon$ for \eqref{test}. This can be seen as an alternative to the $\chi^2$ test for independence which handles covariates.
In practice we do not have access to the distribution under the null $\mathbb{P}_{\mathbf{H}_0}\left(\lambda_0(Y)<\lambda\right)$,
but perform parametric bootstrap \citep{efron1979bootstrap} to compute a proxy $\check{\lambda}_{\varepsilon}$.
In practice we recommend $\varepsilon=0.05$ and use $\lambda_{\textsf{QUT}}:=\check{\lambda}_{.05}$, and refer to it in what follows as \textit{quantile universal threshold} (QUT). This selection of the regularization parameter is essentially the universal threshold of \citet{Dono94b} extended to our setting.\\

\section{Simulation study}
\label{experiments}
\subsection{Estimation }
We simulate $Y\in\mathbb{N}^{300\times 30}$ under model \eqref{eq:poisson}--\eqref{eq:x-model2}, with $R\in \mathbb{R}^{500\times 3}$ and $C\in\mathbb{R}^{300\times 4}$ drawn from multivariate Gaussian distributions with mean $0$ and block-diagonal covariance matrices $\Sigma_R$ and $\Sigma_C$ respectively. We set $\mu^* = 1$, $\alpha^* = (2, 0, 0)$, $\beta^* = (-2,0,0,0)$ and $\Theta$ of rank $5$. We compare the performance of LORI in terms of estimation of the regression coefficients $\alpha$ and $\beta$, and compare it to a standard Poisson Generalized Linear Model (GLM) estimated with the \texttt{glm} function in R. We repeat the experiment $100$ times for decreasing values of the ratio $\tau=\norm{\Theta}[F]/\norm{X_0}[F]$, where $X_0 = (\mu + R_i\alpha + C_j\beta)_{ij}$ is fixed. We look at the Root Mean Square Error (RMSE) for the estimation of $X_0$; the results are given in \Cref{tbl:rmse-alpha}, where we observe that LORI and the Poisson GLM are equivalent for $\tau=0$, and that LORI outperforms the GLM for non-zero interactions, with a gap widening as $\tau$ increases.
\begin{table}[ht]
\centering
\begin{tabular}{|c|ccc|}
  \hline
$\tau$& & Mean of RMSE$*100$ & Standard deviation of RMSE$*100$ \\ 
  \hline
\multirow{2}{*}{$1$}& LORI & $52$ & $1.6$ \\ 
&  GLM & $143$ & $29$ \\ 
   \hline
   \multirow{2}{*}{$ 0.5$}& LORI & $17$ & $0.9$ \\ 
&  GLM & $22$ & $1.0$ \\ 
   \hline
   \multirow{2}{*}{$0.25$}& LORI & $4.3$ & $0.37$ \\ 
&  GLM & $5.3$ & $0.33$ \\ 
   \hline
   \multirow{2}{*}{$0.1$}& LORI & $1.7$ & $0.34$ \\ 
& GLM & $2.6$ & $0.34$ \\ 
   \hline
      \multirow{2}{*}{$0$}& LORI & $0.88$ & $0.2$ \\ 
&  GLM & $0.87$ & $0.19$ \\ 
   \hline
\end{tabular}
\caption{Estimation error (RMSE) of regression coefficients $\sqrt{\norm{\hat\alpha-\alpha^*}[2]^2+\norm{\hat\beta-\beta^*}[2]^2}$ of LORI and a Poisson GLM, for decreasing values of $\tau = \norm{\Theta}[F]/\norm{X_0}[F]$.}
\label{tbl:rmse-alpha}
\end{table}

Second, we compare LORI to a convex low-rank matrix estimatiob procedure with a Poisson loss function and where covariates are not modeled (e.g. \citet{Lafond2015}), in terms of the relative estimation error $\norm{\hat X- X^*}[F]/\norm{X^*}[F]$ (because $\norm{X}[F]$ varies with $\tau$). We refer to this competitor as "Poisson LRM". Again, we reproduce the experiment $100$ times for decreasing values of the ratio $\tau=\norm{\Theta}[F]/\norm{X_0}[F]$. 
\begin{table}[ht]
\centering
\begin{tabular}{|c|ccc|}
  \hline
$\tau$& & Mean of Relative RMSE$*100$ & Standard deviation of Relative RMSE$*100$ \\ 
  \hline
\multirow{2}{*}{$1$}& LORI & $82$ & $1.8$ \\ 
&  Poisson LRM & $95$ & $4.4$ \\ 
   \hline
   \multirow{2}{*}{$0.5$}& LORI & $40$ & $0.25$ \\ 
&  Poisson LRM & $50$ & $0.17$ \\ 
   \hline
   \multirow{2}{*}{$0.25$}& LORI & $24$ & $0.12$ \\ 
&  Poisson LRM & $40$ & $0.12$ \\ 
   \hline
   \multirow{2}{*}{$0.1$}& LORI & $11$ & $0.081$ \\ 
&  Poisson LRM & $36$ & $1.5$ \\ 
   \hline
      \multirow{2}{*}{$0$}& LORI & $4.3$ & $0.1$ \\ 
&  Poisson LRM & $34$ & $1.2$ \\ 
   \hline
\end{tabular}
\caption{Estimation error (Relative RMSE) of parameter matrix $\norm{\hat X-X^*}[F]/\norm{X^*}[F]$ of LORI and a Poisson GLM, for decreasing values of $\tau = \norm{\Theta}[F]/\norm{X_0}[F]$.}
\label{tbl:rmse-x}
\end{table}
On \Cref{tbl:rmse-x}, we observe that LORI achieves lower errors than Poisson LRM, which is expected as we simulated under the LORI model. As $\tau$ decreases -- i.e. the size of the main effects increases relative to the interactions -- both errors decrease as well, and the gap between LORI and the Poisson LRM widens, indicating that modeling covariates explicitly improves the estimation.

\subsection{Imputation}
Using the same simulation scheme, we now compare LORI in terms of missing values imputation to Correspondence Analysis (CA) and Trends \& Indices for Monitoring data \citet{trim} (TRIM), a method based on a Poisson log-linear model used to impute bird abundance data. To do so we erase an increasing proportion of entries in the data and impute them using LORI, CA and TRIM, replicating the experiment 100 times. We also impute the missing values using the column means, as a baseline referred to as "MOY". We observe on \Cref{impute-ca-lori-trim} that LORI performs best, which is expected as we simulate under the LORI model. Moreover, the gap widens as the percentage of missing values increases. In particular, the error of TRIM for 80\% of missing entries is not represented because the method fails (we use the default parameters of the R package \texttt{rtrim}).
\begin{figure}[H]
  \centering
      \begin{subfigure}[b]{0.4\textwidth}
            \centering
            \includegraphics[scale=0.35]{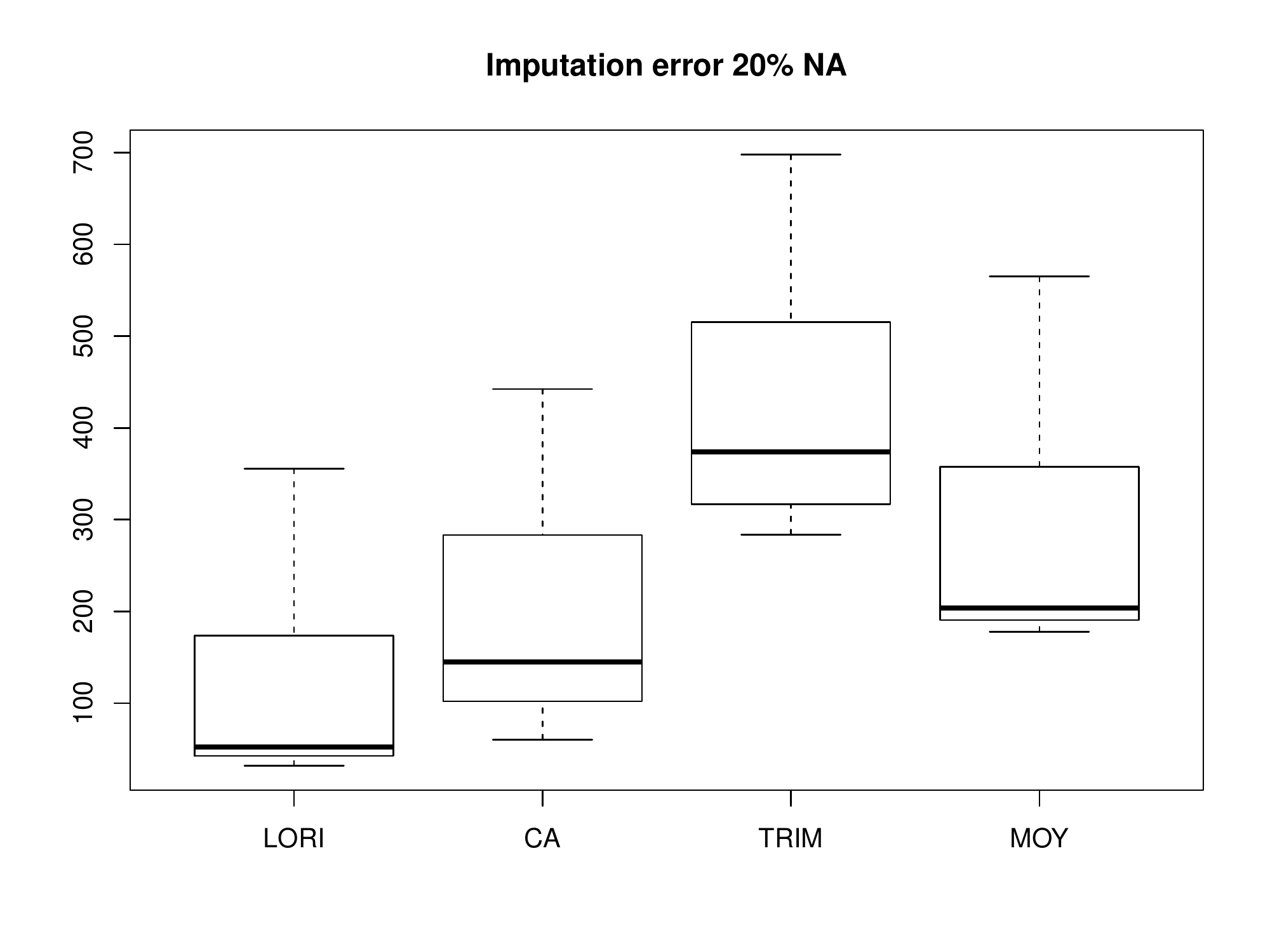}
            \caption{20\% of missing entries}
    \end{subfigure}
\hspace{1cm}
       \begin{subfigure}[b]{0.4\textwidth}
            \centering
            \includegraphics[scale=0.35]{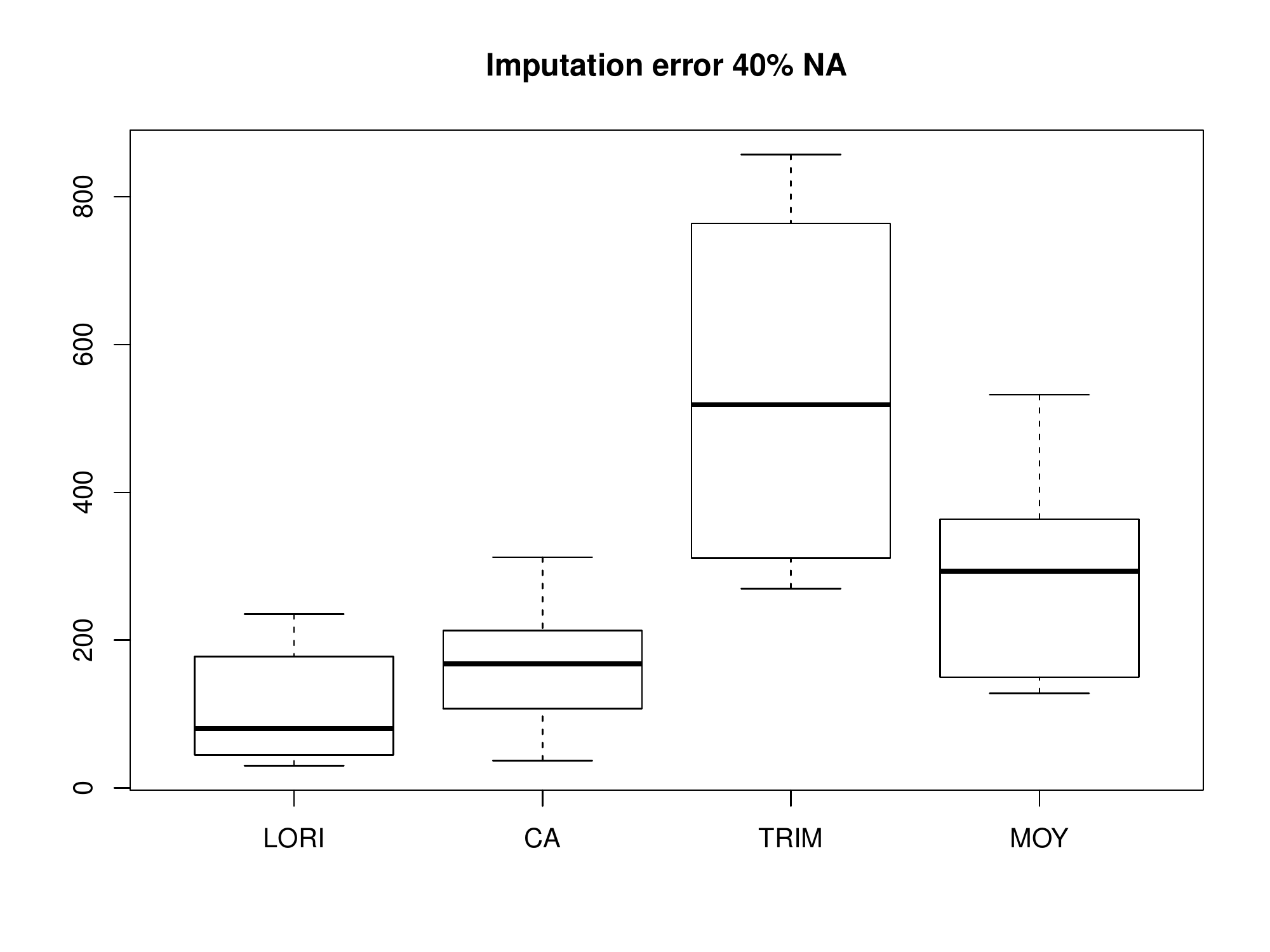}
            \caption{40\% of missing entries}
    \end{subfigure}
         \begin{subfigure}[b]{0.4\textwidth}
            \centering
            \includegraphics[scale=0.35]{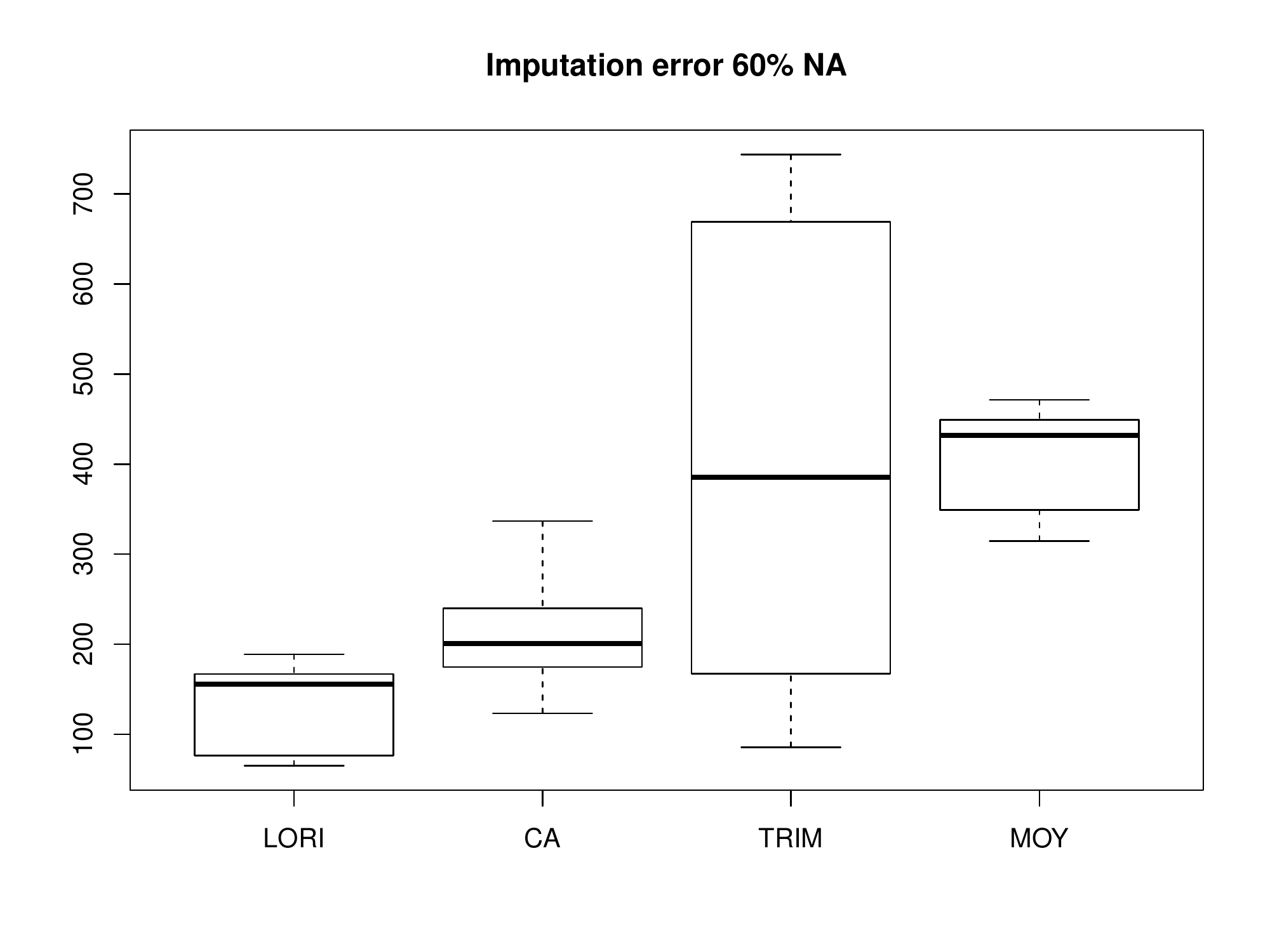}
            \caption{60\% of missing entries}
    \end{subfigure}
\hspace{1cm}
       \begin{subfigure}[b]{0.4\textwidth}
            \centering
            \includegraphics[scale=0.35]{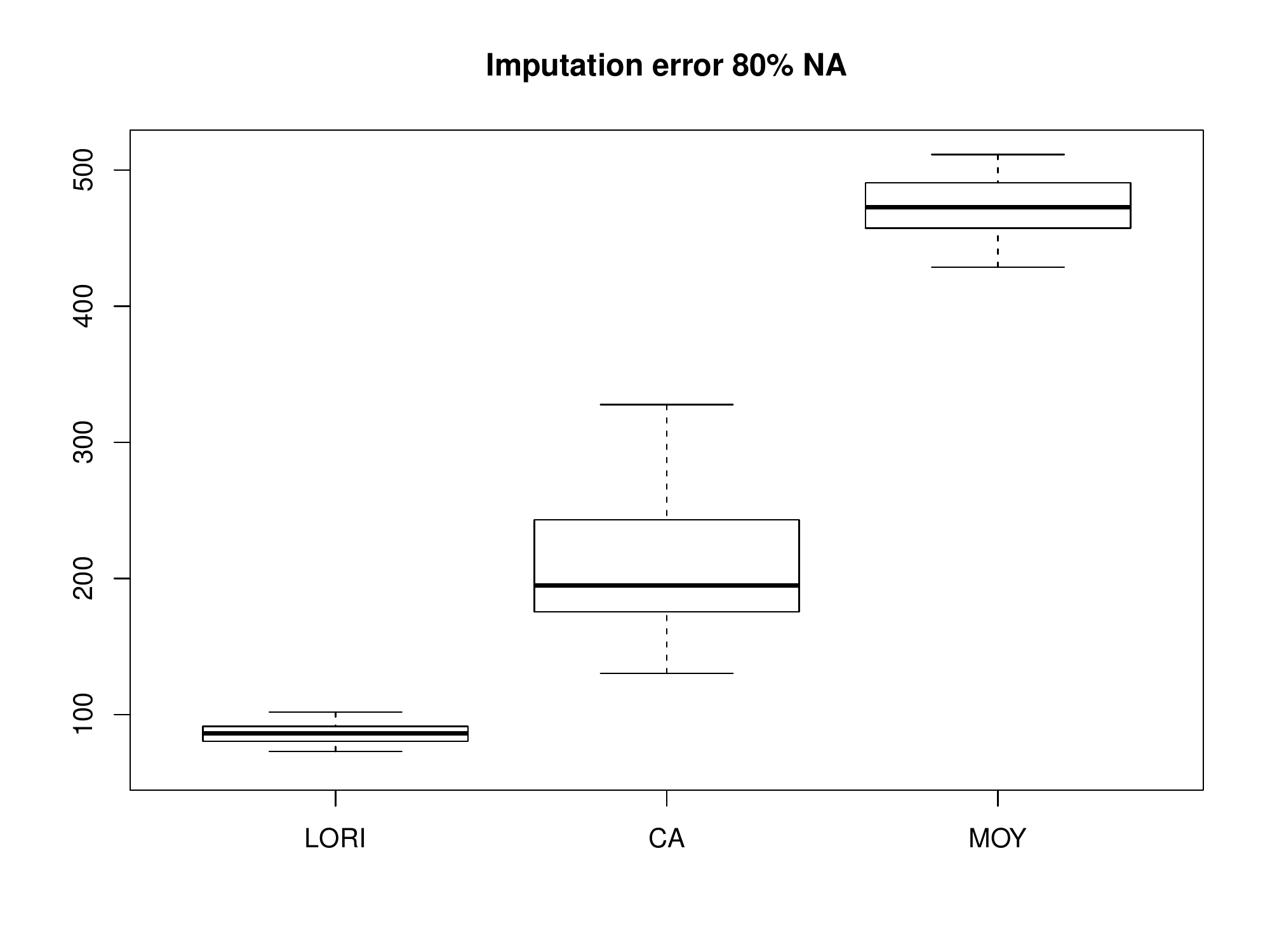}
            \caption{80\% of missing entries}
    \end{subfigure}
    \caption{Average imputation error $\sum_{(i,j)\in \Omega} (Y_{ij}-\hat{Y}_{ij})^2/|\Omega|$.}
  \label{impute-ca-lori-trim}
\end{figure}

\section{Analysis of the Aravo data}
\label{real-data}
The Aravo data set \citep{aravo} counts the abundance of $82$ species of alpine plants in $75$ sites in France; covariates about the environments and species are also available. We focus on $8$ species traits providing physical information about plants (height, spread, etc.), and $4$ environmental variables giving geographical and meteorological information about sites. We apply our method LORI after scaling the covariates and tuning the regularization parameter with the QUT method. This results in estimates for the main effects of the environment characteristics $\alpha$ and of the species traits $\beta$, and of the interaction matrix $\Theta$.
\begin{table}[H]
\centering
\begin{tabular}{rrrrr}
  \hline
Aspect & Slope & PhysD & Snow \\
  \hline
0.04 & 0.07 & -0.02 & -0.07 \\
   \hline
\end{tabular}
\caption{Main effect of the Aravo environment characteristics estimated with LORI. The regularization parameter is tuned using QUT.}
\label{tbl:mainrow}
\end{table}
\begin{table}[H]
\centering
\begin{tabular}{rrrrrrrrr}
  \hline
Height & Spread & Angle & Area & Thick & SLA & Nmass & Seed \\
  \hline
0.09 & -0.24 & -0.18 & -0.20 & -0.11 & -0.17 & 0.18 & -0.12 \\
   \hline
\end{tabular}
\caption{Main effect of the Aravo species traits estimated with LORI. The regularization parameter is tuned using QUT.}
\label{tbl:maincol}
\end{table}
The main effects of environment characteristics are given in \Cref{tbl:mainrow} and the main effects of the species traits in \Cref{tbl:maincol}. First we observe that overall, species traits have larger effects than environment characteristics on the observed abundances. In particular, the mass-based leaf nitrogen content (Nmass) has a large positive effect, which seems to indicate that plants with a large Nmass tend to be more abundant across all environments. On the other hand, the maximum lateral spread of clonal plants (Spread), area of single leaf (Area) leaf elevation angle estimated at the middle of the lamina (Angle) and specific leaf area (SLA) have large negative effects on the abundances.\\

The estimated rank of the interaction matrix $\hat\Theta$ (number of singular values above $10^{-6}$) is $2$. The environments (rows) and species (columns) can be visualized on a biplot \cite[Section 2.5]{deRooij2005}, where rows and columns are represented  simultaneously in a normalized Euclidean space. In such plots, the dimensions of the Euclidean space are given by the principal directions of $\hat\Theta$, scaled by the square root of the singular values of $\hat\Theta$. Figure \ref{aravo-biplot} shows such a display, which can be interpreted in terms of distance between points: a species and an environment that are close interact highly, and two species or two environments that are close have similar profiles. Justifications for such a distance interpretation can be found in \cite[Section 2.5]{deRooij2005} or \cite[Section 2]{JosseFithian2016}.
\begin{figure}
\begin{center}
\includegraphics[scale=0.5]{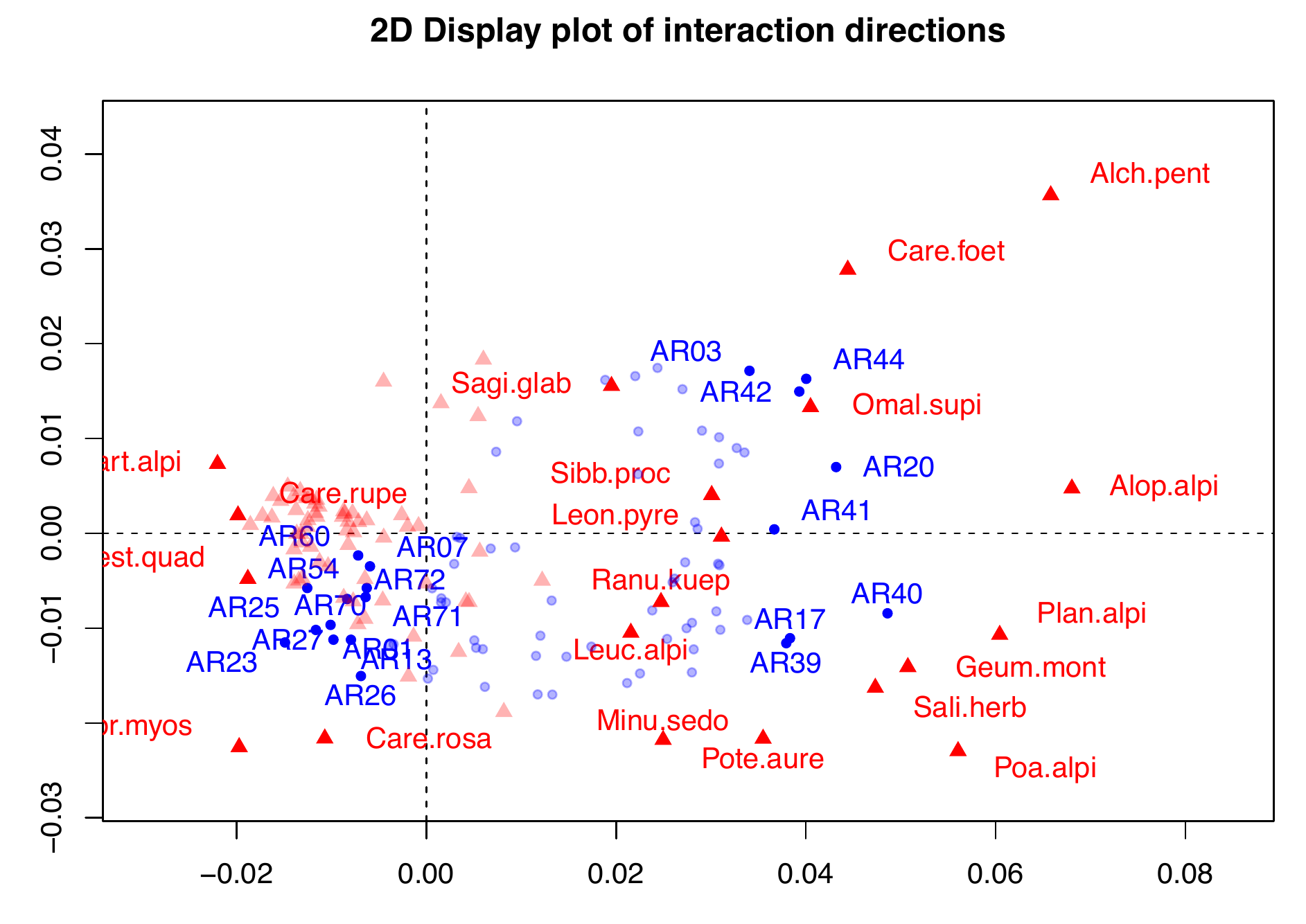}
 \caption{Display of the two first dimensions of interaction estimated with LORI. Environments are represented with blue points and species with red triangles.}
 \label{aravo-biplot}
\end{center}
\end{figure}
We can then look at the relations between the known traits and the interaction directions of $\hat \Theta$.
\Cref{spe-cov} shows that the two first directions of interaction are correlated with the species covariates; the correlation is particularly high for the Nmass and SLA variables. Thus, on \Cref{aravo-biplot}, the two directions separate the plants with large SLA and Nmass (top right corner) from those with small SLA and Nmass (bottom left corner). Then, \Cref{spe-cov} shows that the directions of interaction are also correlated with the environment covariates, and particularly with the mean snowmelt date (Snow). Thus, on \Cref{aravo-biplot}, the two directions separate the late melting environments (top right corner) from the early melting environments (bottom left corner).  
\begin{figure}
  \centering
      \begin{subfigure}[b]{0.4\textwidth}
            \centering
            \includegraphics[scale=0.5]{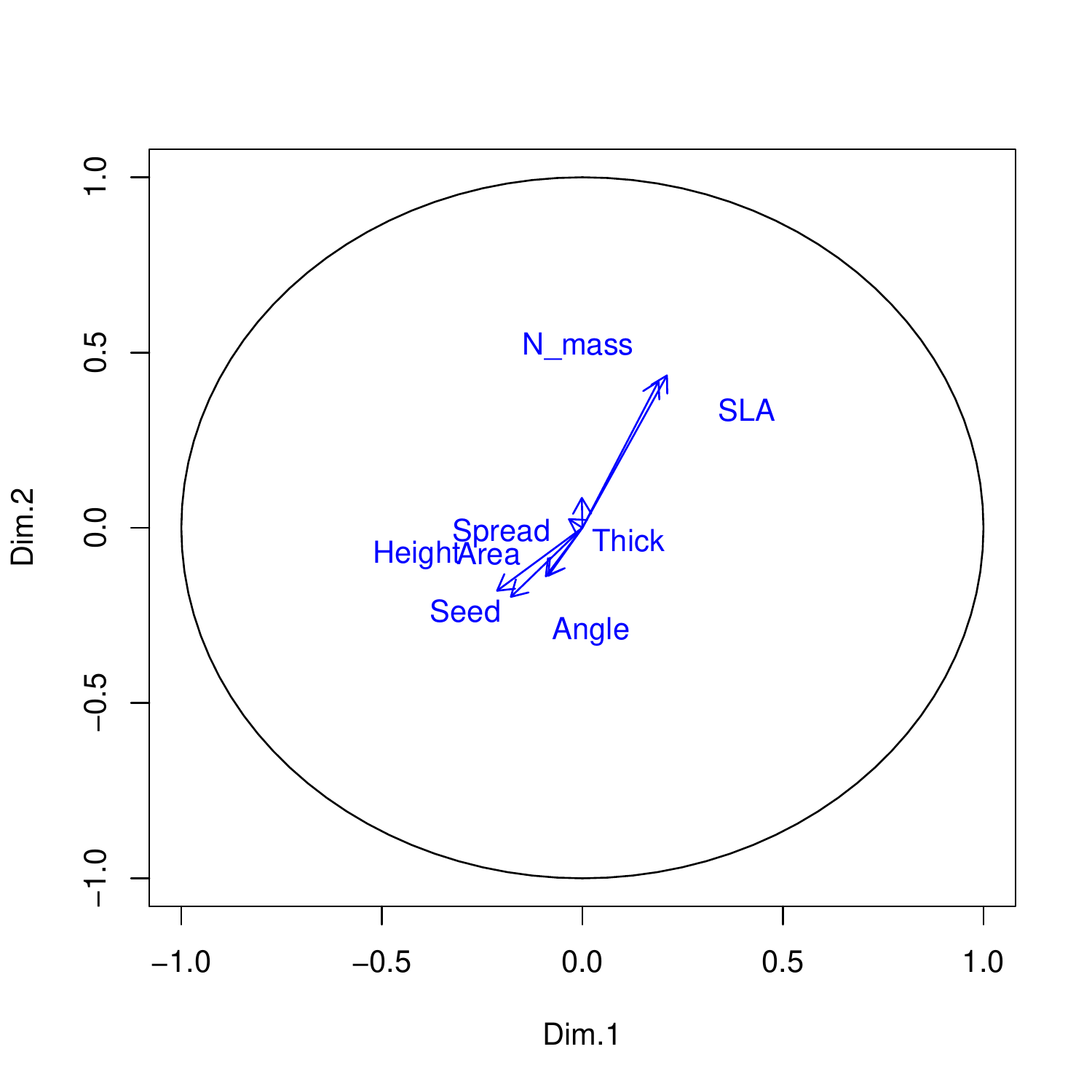}
            \caption{Environment covariates}
            \label{spe-cov}
    \end{subfigure}
\hspace{2cm}
       \begin{subfigure}[b]{0.4\textwidth}
            \centering
            \includegraphics[scale=0.5]{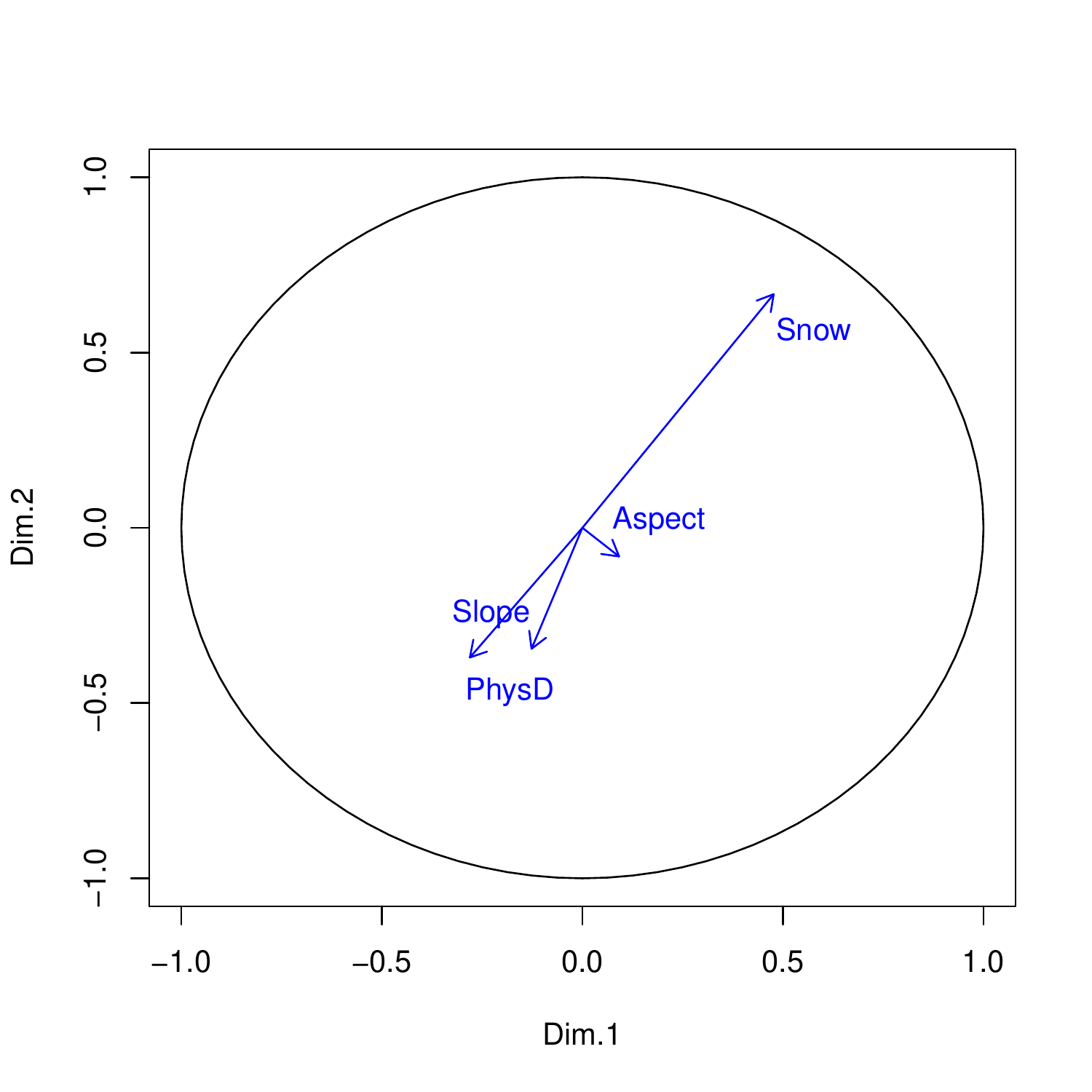}
            \caption{Species traits}
            \label{env-cov}
    \end{subfigure}  \caption{Correlation between the two first dimensions of interaction and the covariates (the covariates are not used in the estimation).}
  \label{aravo-cor}
\end{figure}
Combining the interpretation of \Cref{aravo-biplot}, \Cref{spe-cov} and \Cref{env-cov}, we deduce that plants with large Nmass and SLA interact highly with late melting sites (large value of Snow). This was in fact the main result obtained in the original study \citet{aravo} (see, e.g., the summary of findings in the abstract), which advocates the good properties of LORI in terms of interpretation. \\

\section{Using covariates to impute ecological data}
\label{eco-data}

The water-birds data count the abundance of migratory water-birds in 785 wetland sites (across the 5 countries in North Africa),  between 1990 and 2017 \citep{SAYOUD201711}. One of the objectives is to assess the effect of time on species abundances, to monitor the populations and assess wetlands conservation policies. Ornithologists have also recorded side information concerning the sites and years, which may influence the counts. For instance, meteorological anomalies, latitude and longitude.
The count table contains a large amount of missing entries (70\%), but the covariate matrices which contain respectively 6 covariates about the 785 sites and 8 covariates about the 18 years, are fully observed.
Our method allows to take advantage of the available covariates to provide interpretation for spatio-temporal patterns. As a by-product, it produces an imputed contingency table. \\

\Cref{tbl:main-site-mil,tbl:main-year-mil} show the estimated main effects of some of the sites and years characteristics. Sites with large latitudes are associated to smaller counts, as well as sites which are far from the coast. Sites which are located far from towns, and sites with large water surfaces are associated to larger counts. The four year covariates given in \Cref{tbl:main-year-mil} concern meteorological anomalies. The associated coefficients are all negative, indicating that more important anomalies are associated to smaller abundances.

\begin{table}
\centering
\begin{tabular}{rrrrr}
  \hline
Agricultural surface & Latitude & Dist. to town & Dist. to coast & Surface \\
  \hline
0.09 &-0.21 & 0.09 & -0.48 & 0.20 \\
   \hline
\end{tabular}
\caption{Main effect of the sites characteristics estimated with LORI. The regularization parameter is tuned using QUT.}
\label{tbl:main-site-mil}
\end{table}

\begin{table}
\centering
\begin{tabular}{rrrrrrrr}
  \hline
Spring N/O & Spring N/E & Winter S/O& Winter S/E\\
  \hline
-0.05 & -0.03 & -0.05 & -0.03 \\
   \hline
\end{tabular}
\caption{Main effect of the years characteristics estimated with LORI. }
\label{tbl:main-year-mil}
\end{table}
The sites and years can also be displayed using the same visual tools as described in \Cref{real-data}. \Cref{fig:waterbirds-cov-sites} and \ref{fig:waterbirds-cov-years} show the correlations between the covariates and the directions of interaction. On the two-dimensional display on \Cref{fig:waterbirds-scatter}, the first dimension is correlated with geographical characteristics, and the second dimension with meteorological anomalies. We observe a very clear temporal gradient along the second dimension, indicating that over time, meteorological anomalies increase (in the sense of a summary anomaly variable embodied by the second direction). One of the site (378) lays out of the point cloud, and corresponds to a site with very large surface. \\
\begin{figure}[!ht]
\centering
\begin{subfigure}[t]{0.45\textwidth}
\centering
\includegraphics[scale=0.4]{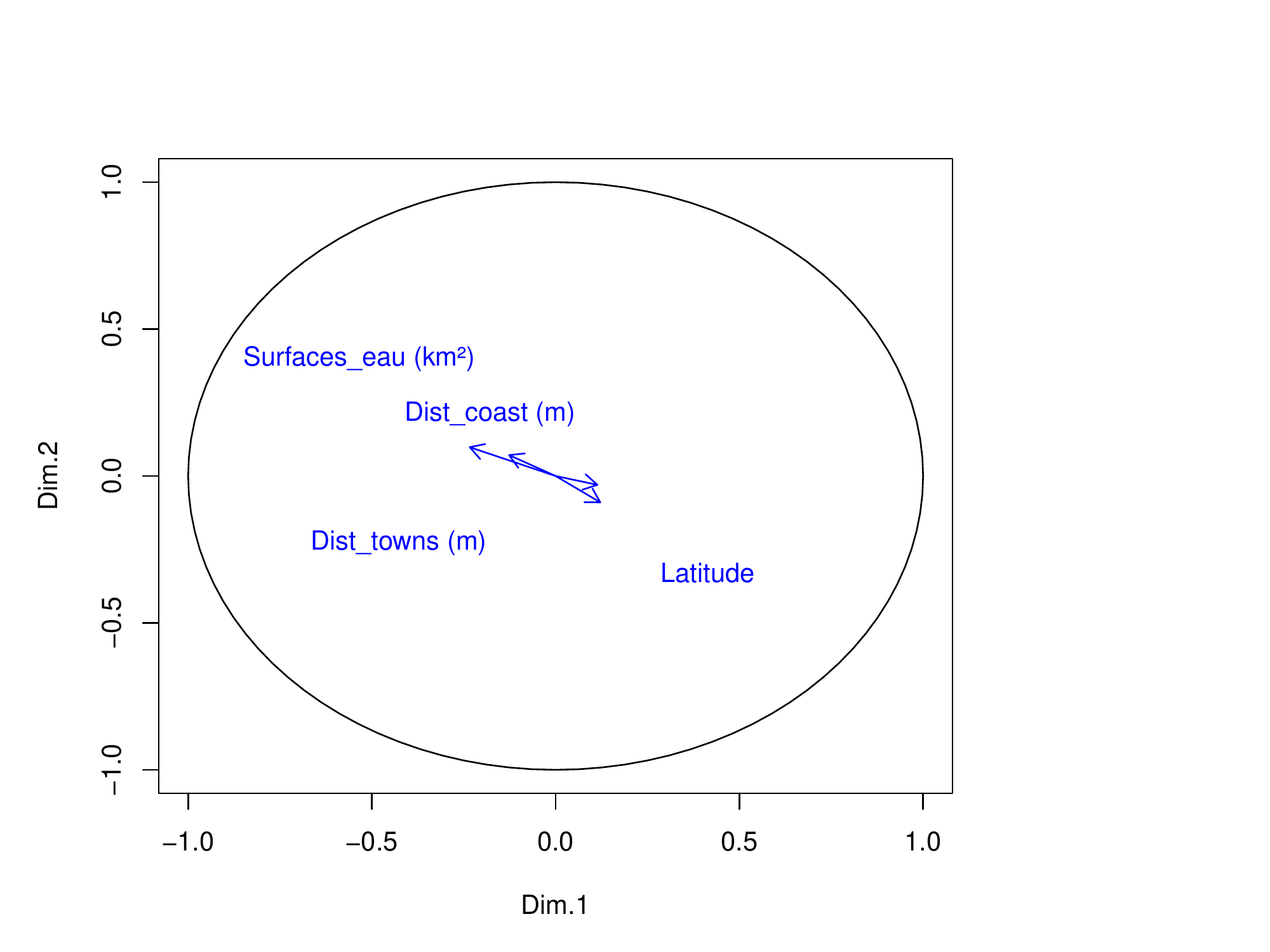}
\caption{Correlation between sites characteristics and the two first directions of interaction.}
\label{fig:waterbirds-cov-sites}
\end{subfigure}
\hspace{0.8cm}
\begin{subfigure}[t]{0.45\textwidth}
\centering
\includegraphics[scale=0.4]{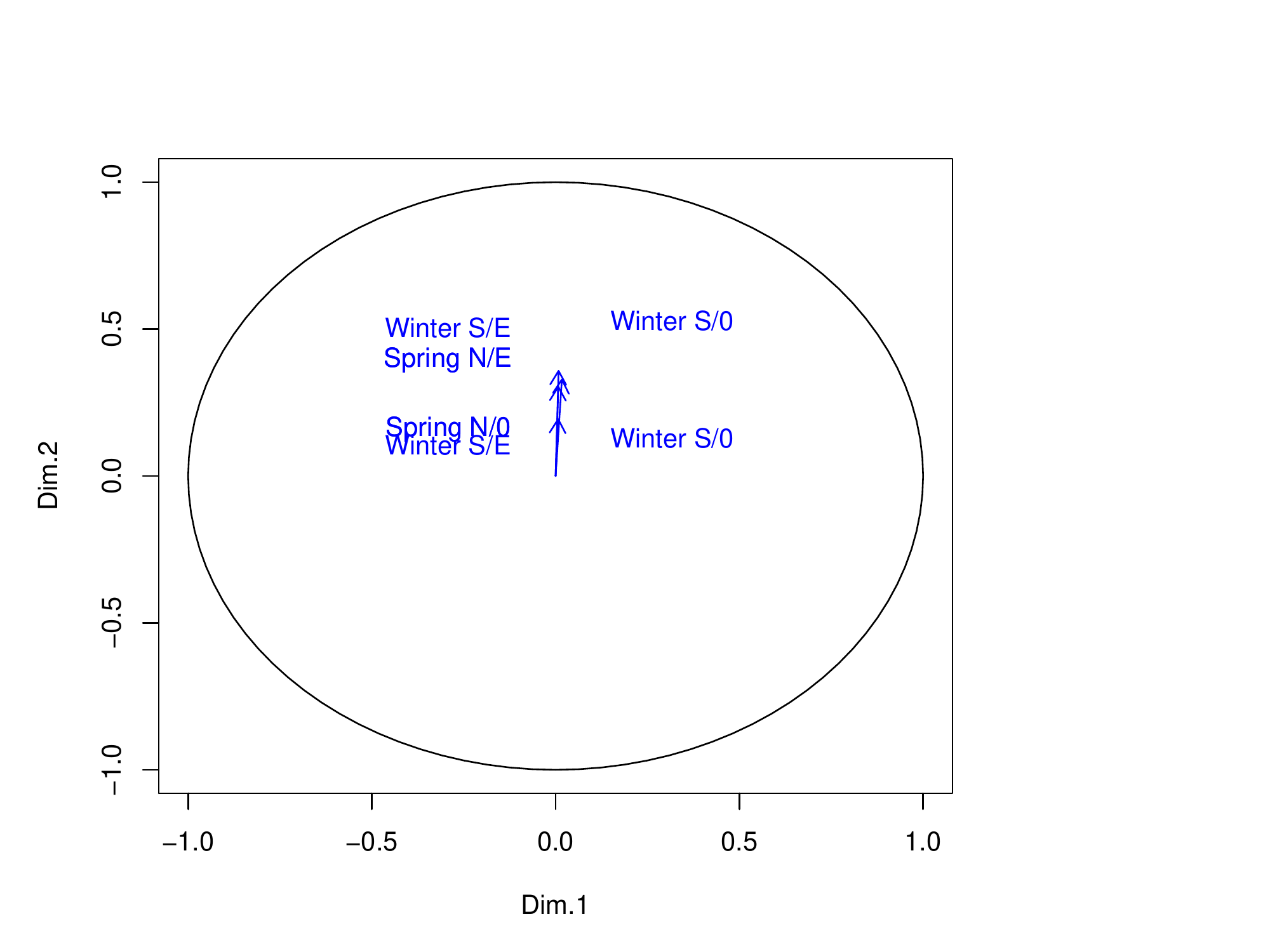}
\caption{Correlation between year characteristics and the two first directions of interaction.}
\label{fig:waterbirds-cov-years}
\end{subfigure}
\begin{subfigure}[t]{0.8\textwidth}
\centering
\includegraphics[scale=0.5]{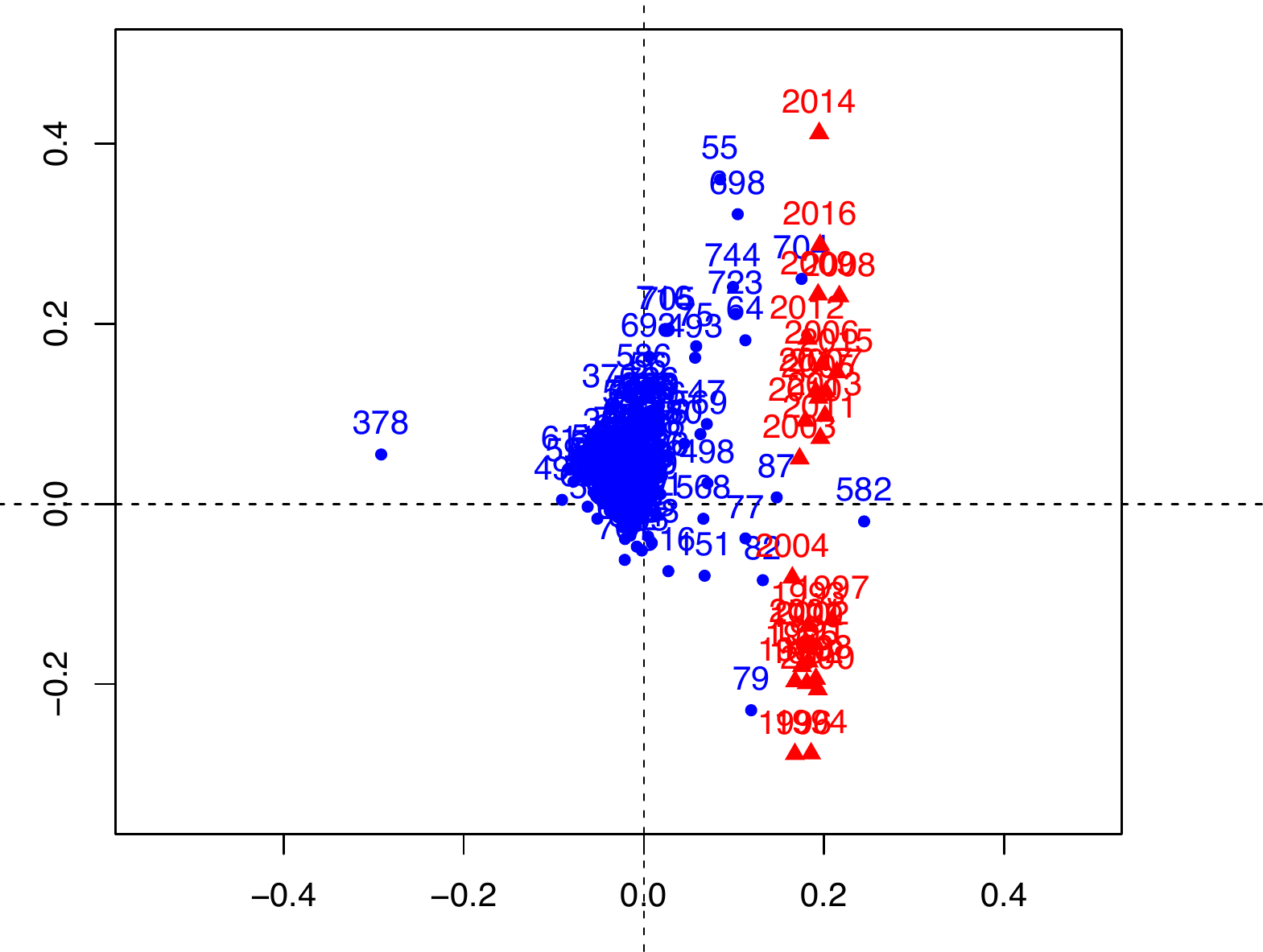}
\caption{Display of the two first dimensions of interaction estimated with LORI. Environments are represented with blue points and years with red triangles.}
\label{fig:waterbirds-scatter}
\end{subfigure}
\caption{Visual display of LORI results for the water-birds data.}
\label{fig:waterbirds-decomp}
\end{figure}

LORI also returns counts estimates, which can be used to compute an estimation of the total yearly abundances (i.e. counts estimates summed across sites). To better assess the temporal trend, one can decompose the estimated counts into three factors corresponding to the site effects, year effects and interactions respectively. Indeed, for $(i,j)\in\{1,\ldots,n\}\times\{1,\ldots,p\}$, one can write $$\exp(\hat X_{ij}) = \exp(\hat\mu)\exp(R_{i,.}\hat\alpha)\exp(C_{j,.}\hat\beta)\exp(\hat\Theta_{ij}).$$ \Cref{fig:waterbirds-decomp} shows the last three factors of this decomposition separately. \\
\begin{figure}[!ht]
\centering
\begin{subfigure}[t]{0.45\textwidth}
\centering
\includegraphics[scale=0.4]{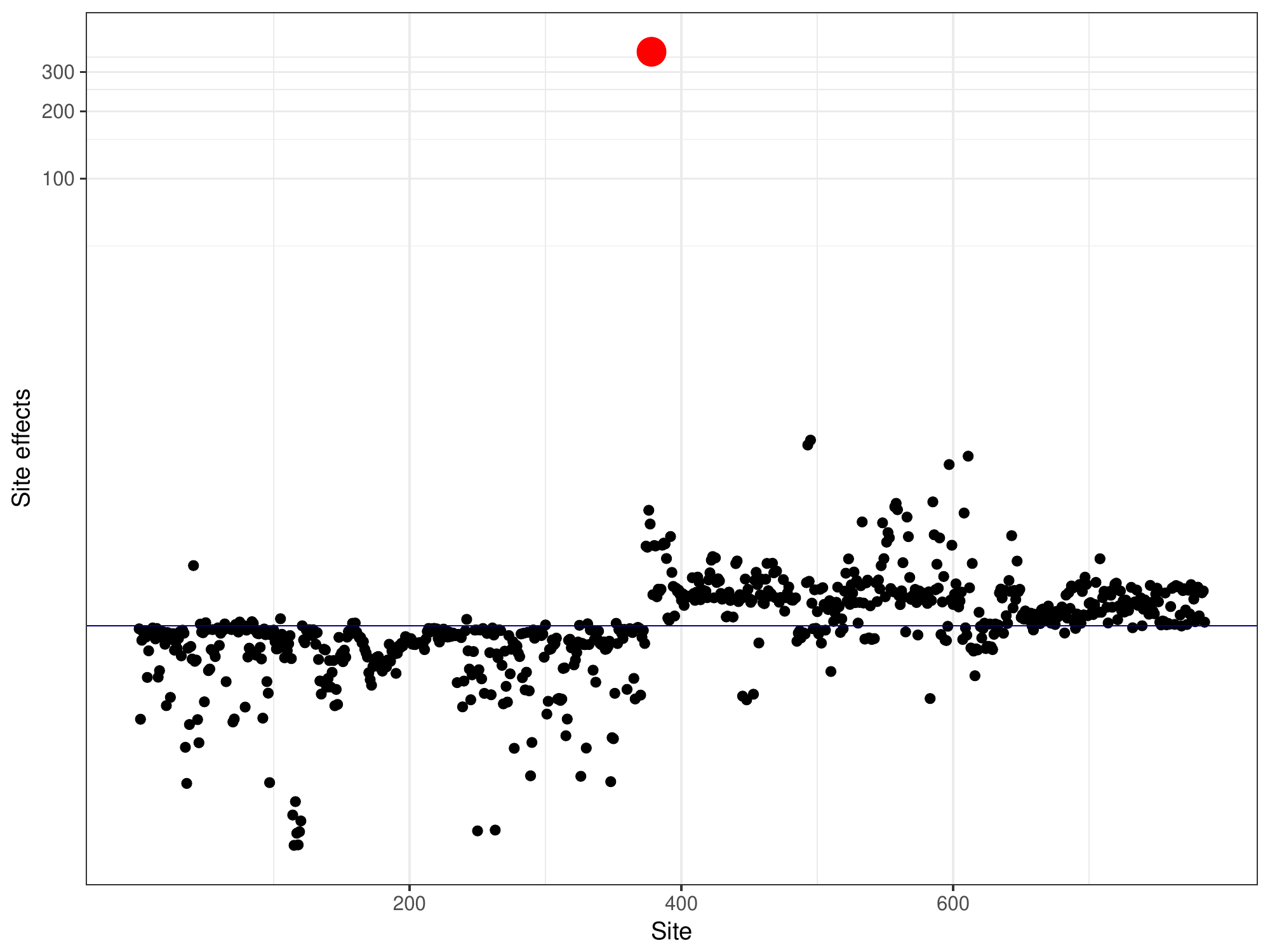}
\caption{Total site effects in count scale ($\exp(R_{i,}\hat\alpha)$, for $1\leq i\leq n$). One site has $\exp(R_{i,}\hat\alpha)\geq \exp(5)$ (large red point). Horizontal line: $\exp(R_{i,}\alpha)= 1$.}
\label{fig:waterbirds-decomp-sites}
\end{subfigure}
\hspace{0.8cm}
\begin{subfigure}[t]{0.45\textwidth}
\centering
\includegraphics[scale=0.4]{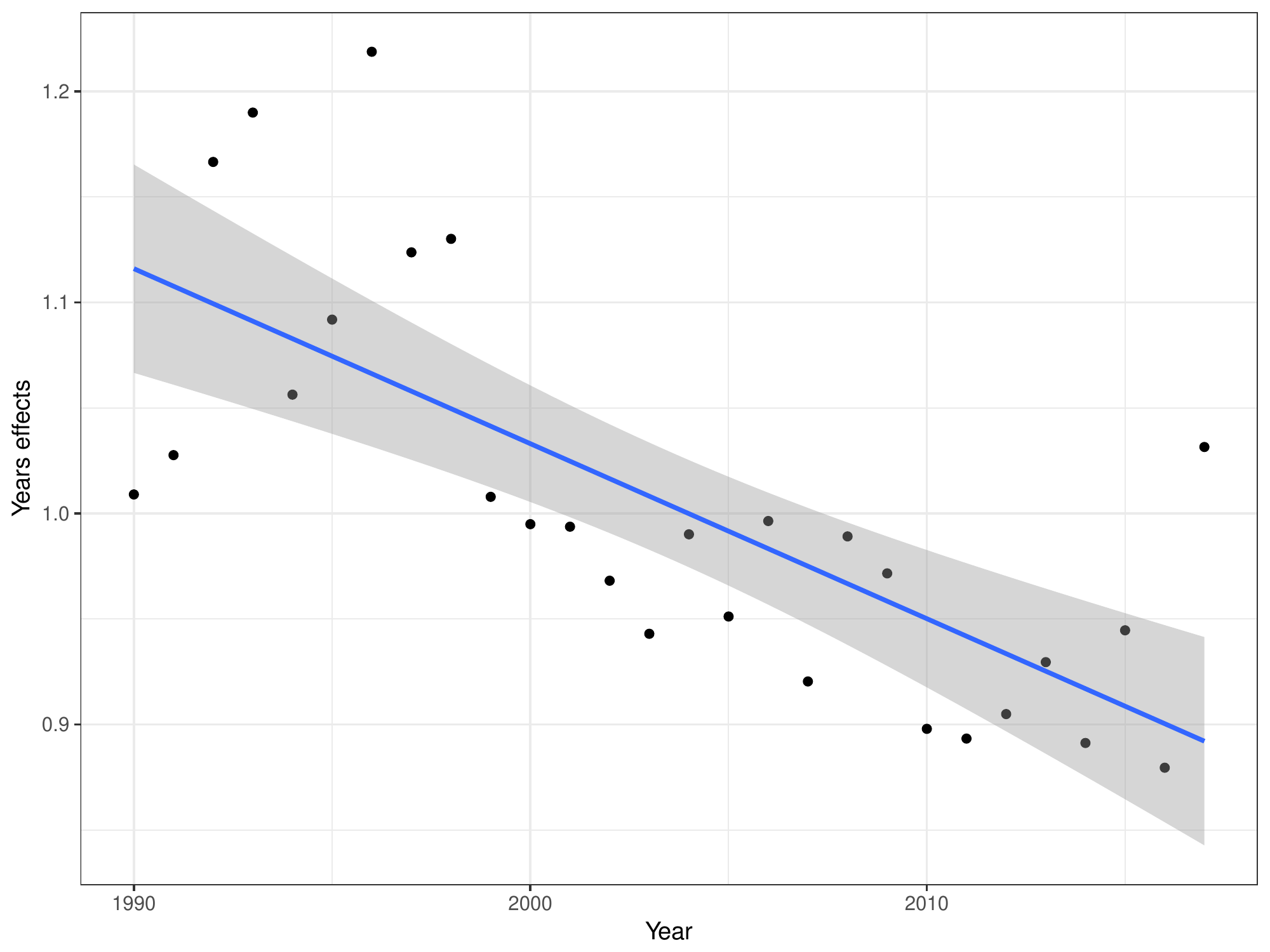}
\caption{Total year effects in count scale ($\exp(C_{j,}\hat\beta)$, for $1\leq j\leq p$). Blue line: loess (standard deviation in gray).}
\label{fig:waterbirds-decomp-years}
\end{subfigure}
\begin{subfigure}[t]{0.8\textwidth}
\centering
\includegraphics[scale=0.5]{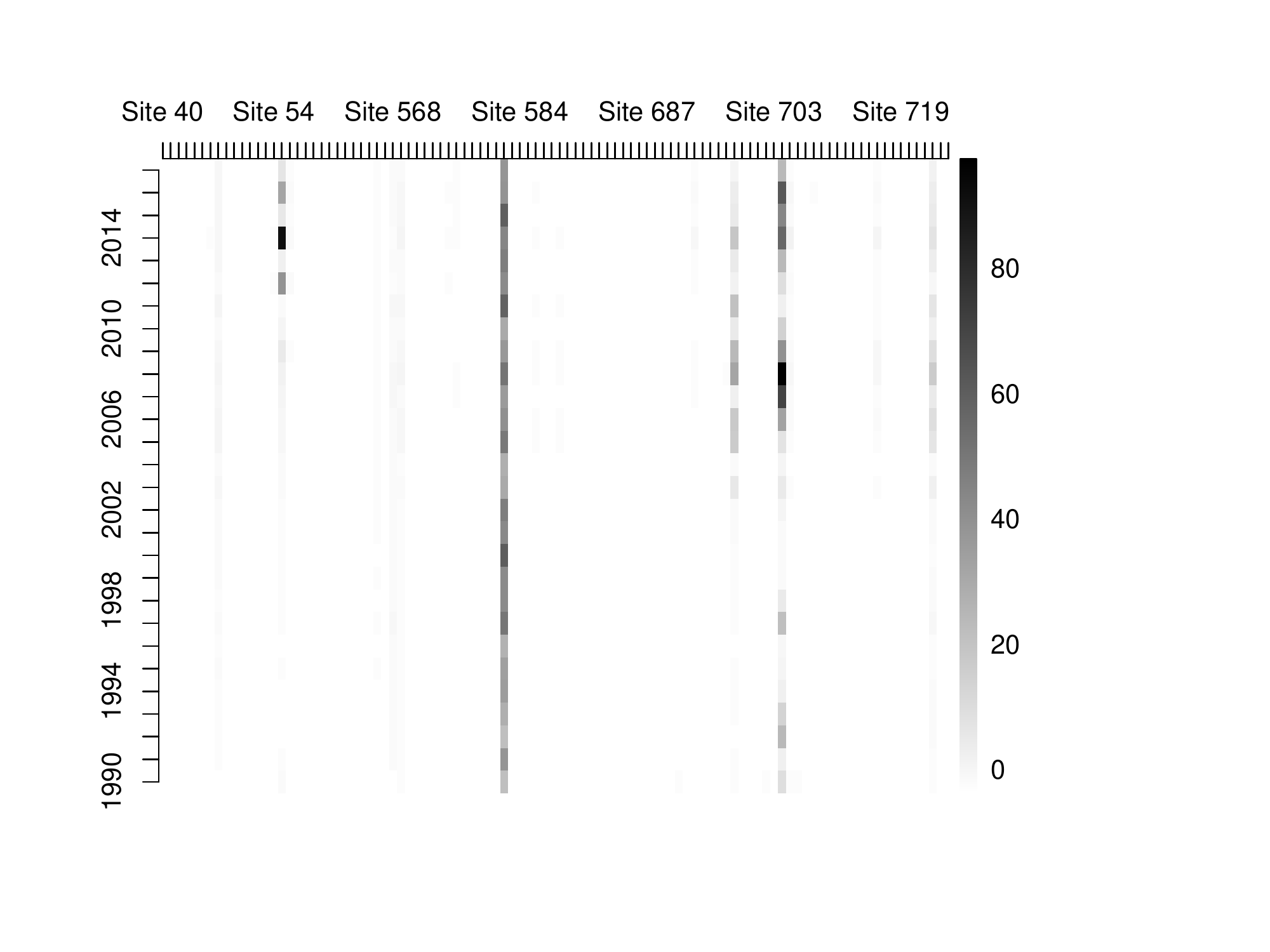}
\caption{Interaction in count scale ($\exp(\hat\Theta_{ij})$) for $30$ sites (to improve the display, the rest of the sites take value extremely close to $1$). }
\label{fig:waterbirds-decomp-inter}
\end{subfigure}
\caption{Decomposition of the estimated counts into multiplicative site effects (top left), year effects (top right) and interactions (bottom).}
\label{fig:waterbirds-decomp}
\end{figure}

On \Cref{fig:waterbirds-decomp-sites} we see that most sites have multiplicative effects around $1$ on count scale. One site (site 378, large red point) stands out; again, it corresponds to an extremely large site ($6000$km$^2$, $5$ times larger than the second, $300$ times larger than the mean). In this respect, the row effects act as normalization factors accounting for surface. We also observe tenuous levels along the $x$ axis, corresponding to sites of different countries. On \Cref{fig:waterbirds-decomp-years} we observe a decreasing temporal trend. This means that, all other things being equal, later years tend to produce smaller abundances. As illustrated in \Cref{fig:waterbirds-cov-years}, this temporal trend can be associated with the effects of meteorological anomalies. Note that the temporal effects (top right) are smaller in amplitude than the spatial effects (top left). Indeed, more variability is observed between sites in a given year, than between years for a given site.\\

Finally, looking at the interaction matrix on \Cref{fig:waterbirds-decomp-inter}, we see that the interaction is mainly driven by a few sites which interact more or less highly with every year. In particular the site $582$ presents large interactions with every year, and corresponds to a national park in Morocco and is the most abundant site ($4$ million birds in total, twice as much as the second most abundant, $120$ more than the mean). Here, the large abundance is not explained by the geographical covariates available (but maybe by other unmeasured factors such as protection legislations), and thus the extreme abundance is captured in the interaction rather than the main effects. This profile was also visible on \Cref{fig:waterbirds-scatter} where the site $582$ lies amongst the cloud of years, indicating large interactions through small Euclidean distances.\\
 The site 704 also presents large interactions, but they are not constant throughout the years. It corresponds to Ichkeul national park in Tunisia, which is a major site for most species; the abundances are very large in Ichkeul compared to other sites. However during several years including 2007, bad weather conditions prevented ornithologist to correctly count the birds, thus reported counts are significantly lower than expected. This explains the drop in the interaction in 2007 for Ichkeul, corresponding to an outlier behavior. Again, such a profile could not be highlighted without modeling interactions. This illustrates one of the advantages of LORI for such bird abundance data compared to state-of-the-art methods such as \citet{trim} which do not model interactions. In particular, in most cases the interaction terms absorb outlying values (small or large), and indirectly account for the over-dispersion which is known to occur in birds abundance data.

\section{Discussion}
\label{discussion}
We conclude by discussing some opportunities for further research.
 To select covariates, we could  penalize the main effects with an $\ell_1$ penalty on $\alpha$ and $\beta$. It may be also of interest to consider other sparsity inducing penalties. In particular, penalizing the Poisson log-likelihood by the absolute values of the coefficients of the interaction matrix $\Theta$ could possibly lead to solutions where some interactions are driven to $0$ and a small number of large interactions are selected.
Secondly, it would be useful to develop a multiple imputation procedure based on LORI, to provide confidence regions for the estimated parameters.
The properties of the thresholding test, which can be seen as an alternative to a chi-squared test for independence with covariates, also merit further investigation. In particular, the power could be assessed.
Finally, we could also explore whether our model could be extended to more complex models such as the zero-inflated negative binomial models.

\section*{Acknowledgments}
The authors thank Trevor Hastie, Edgar Dobriban, Olga Klopp, Kevin Bleakey and St\'{e}phane Dray for their very helpful comments on this manuscript. We would like to thank all the organizations involved in the Mediterranean Waterbirds Network (MNW) who provided the waterbird data set used in this article: the GREPOM/BirdLife Morocco, the Direction Générale des Forêts (Algeria), the  AAO/BirdLife Tunisia, the Libyan Society for Birds, the Egyptian Environment Affairs Agency, the national agency for wildlife and hunting management (ONCFS, France) and the Research Institute of the Tour du Valat (France). We are also grateful to all the field observers who participated in the North African IWC thus making this data set so rich, and to Pierre Defos du Rau and Laura Dami for their help.

\section{Proofs}
\subsection{Proof of \Cref{th:global-bound-1}}
\label{proof:global-bound-1}
We will first derive an upper bound for $\sum_{(i,j)\in\Omega}(\hat X_{ij} - X_{ij}^*)^2$, then control $\norm{\hat X - X^*}[F]^2$ by $\norm{\hat X - X^*}[F]^2\leq \sum_{(i,j)\in\Omega}(\hat X_{ij} - X_{ij}^*)^2 + \mathsf{D}$, with $\mathsf{D}$ a residual term defined later on. By definition of $\hat{X} = \hat{X_0} + \hat\Theta$, $
\mathcal{L}(\hat{X}) + \lambda\norm{\hat\Theta}[*]\leq \mathcal{L}(X^*) + \lambda\norm{\Theta^*}[*].$ Using the strong convexity of $\mathcal{L}$ and substracting $\langle\nabla\mathcal{L}(X^*),\hat X - X^*\rangle$ on both sides of this inequality, we obtain
\begin{equation}
\label{eq:first-bound}
\frac{\sigma_{-}^2\sum_{(i,j)\in\Omega}(\hat X_{ij}-X^*_{ij})^2}{2}\leq \underbrace{-\langle\nabla\mathcal{L}(X^*),\hat X - X^*\rangle}_{I}+\underbrace{\lambda(\norm{\Theta^*}[*]-\norm{\hat\Theta}[*])}_{II}.
\end{equation}
We will bound separately the two terms on the right hand side of \eqref{eq:first-bound}. \\

Given a matrix $X\in \mathbb{R}^{n\times p}$, we denote $\mathcal{S}_1(X)$ (resp. $\mathcal{S}_2(X)$) the span of left (resp. right) singular vectors of $X$. Let $P_{\mathcal{S}_1(X)}^{\perp}$ (resp.  $P_{\mathcal{S}_2(X)}^{\perp}$) be the orthogonal projector in $\mathbb{R}^n$ on $\mathcal{S}_1(X)^{\perp}$ (resp. in $\mathbb{R}^p$ on $\mathcal{S}_2(X)^{\perp}$). We define the projection operator in $\mathbb{R}^{n\times p}$ $\mathcal{P}_{X}^{\perp}: \tilde{X} \mapsto P_{\mathcal{S}_1(X)}^{\perp}\tilde{X}P_{\mathcal{S}_2(X)}^{\perp}$, and $\mathcal{P}_{X}: \tilde{X} \mapsto \tilde{X}-P_{\mathcal{S}_1(X)}^{\perp}\tilde{X}P_{\mathcal{S}_2(X)}^{\perp}$. We use the following Lemma, proved in \cite[Lemma 16]{Lafond2015}.
\begin{lemma} \label{lemma:l1}
For all $M$ and $M'$ in $\mathbb{R}^{n\times p}$,
\begin{enumerate}[(i)]
\item \label{lemma:l1:i} $\norm{M+ \mathcal{P}_{M}^{\perp}(M)}[*]=\norm{M}[*] + \norm{\mathcal{P}_{M}^{\perp}(M)}[*],$
\item \label{lemma:l1:ii} $\norm{M}[*]-\norm{M'}[*]\leq  \norm{\mathcal{P}_{M}(M-M') }[*]-\norm{\mathcal{P}_{M}^{\perp}(M-M')}[*],$
\item \label{lemma:l1:iii} $\norm{\mathcal{P}_{M}(M-M')}[*]\leq \sqrt{2\mathrm{rk}(M)}\norm{M-M'}[F].$
\end{enumerate}
\end{lemma}
Using $|\langle\nabla\mathcal{L}(X^*),\hat X - X^*\rangle|\leq\norm{ \hat X - X^* }[*]\norm{ \nabla\mathcal{L}(X^*)}$ and the triangular inequality gives that
\begin{equation}
\label{eq:left-bound}
I \leq  \norm{\nabla\mathcal{L}(X^*)} \left(\norm{ \mathcal{P}_{\Theta^*}(\hat{\Theta}-\Theta^*)}[*]+\norm{\mathcal{P}_{\Theta^*}^{\perp}(\hat{\Theta}-\Theta^*)}[*]+
 \norm{\hat{X}_0-X^*_0}[*]\right).
\end{equation}
Then, \Cref{lemma:l1}~\ref{lemma:l1:ii} applied to $\hat{\Theta}$ and $\Theta^*$, results in
\begin{equation}
\label{eq:right-bound}
II\leq \lambda\left(\|\mathcal{P}_{\Theta^* }(\hat{\Theta} -\Theta^* )  \|_{*}-\|\mathcal{P}_{\Theta^* }^{\perp}(\hat{\Theta} -\Theta^* ) \|_{*}\right).
\end{equation}
Plugging inequalities \eqref{eq:left-bound} and \eqref{eq:right-bound} in \eqref{eq:first-bound} we obtain
\begin{multline}
\label{eq:one}
\sigma_{-}^2\sum_{(i,j)\in\Omega}(\hat X_{ij}-X^*_{ij})^2\leq 2(\lambda + \norm{\nabla\mathcal{L}(X^*)})\left\| \mathcal{P}_{\Theta^*}(\hat{\Theta}-\Theta^*)\right\|_{*}\\
+2(\norm{\nabla\mathcal{L}(X^*)}-\lambda)\norm{\mathcal{P}_{\Theta^* }^{\perp}(\hat{\Theta} -\Theta^* )}[*]+ 2\norm{\nabla\mathcal{L}(X^*)}\norm{\hat X_0-X^*_0}[*].
\end{multline}
We now use the condition $\lambda\geq 2\norm{\nabla\mathcal{L}(X^*)}$ in \eqref{eq:one}:
\begin{equation}
\label{eq:two}
\sigma_{-}^2\sum_{(i,j)\in\Omega}(\hat X_{ij}-X^*_{ij})^2\leq 3\lambda \norm{ \mathcal{P}_{\Theta^*}(\hat{\Theta}-\Theta^*)}[*]
+ \lambda\norm{\hat X_0-X^*_0}[*].
\end{equation}
Then, $\text{rk}(\hat X_0-X^*_0)\leq r$ and $\norm{\hat X_0 - X^*_0}[F]\leq \norm{\hat X - X^*}[F]$ imply that $\norm{\hat X_0-X^*_0}[*]\leq \sqrt{r} \norm{X^*-\hat{X}}[F]$, which together with \Cref{lemma:l1}~\ref{lemma:l1:iii} and $\norm{\hat \Theta - \Theta^*_0}[F]\leq \norm{\hat X - X^*}[F]$ yields
\begin{equation}
\label{eq:lambda-bound}
\sigma_{-}^2\sum_{(i,j)\in\intervi{n}\times\intervi{p}}\omega_{ij}(\hat X_{ij}-X^*_{ij})^2\leq \lambda\left(3\sqrt{2\operatorname{rank}(\Theta^*)}+\sqrt{r}\right)\norm{\hat X-X^*}[F].
\end{equation}
We now derive the upper bound $\norm{\hat X - X^*}[F]^2\leq \sum_{(i,j)\in\intervi{n}\times\intervi{p}}\omega_{ij}(\hat X_{ij} - X_{ij}^*)^2 + \mathsf{D}$. Define $\eta = 72\log(n+p)/(\pi\log(6/5))$,
\begin{equation}
\label{eq:sigma-omega-x}
\Sigma(\omega, X) = \sum_{(i,j)\in\intervi{n}\times\intervi{p}}\omega_{ij}X_{ij}^2
\end{equation}
 and the set
\begin{equation}
\label{eq:set-C}
\mathcal{C}(\eta, \rho) = \left\{X\in\mathbb{R}^{n\times p}; \norm{X}[\infty]\leq 1, \norm{X}[*]\leq \sqrt{\rho}\norm{X}[F], \mathbb{E}\left[\Sigma(\omega, X)\right] >\eta\right\}.
\end{equation}
We start by showing in the following Lemma that whenever $\hat{X} - X^*$ belongs to $\mathcal{C}(\eta,\rho)$ (for $\rho$ and $\mathsf{D}$ defined later on), a restricted strong convexity property of the form $\norm{\hat X - X^*}[F]^2\leq \sum_{(i,j)\in\intervi{n}\times\intervi{p}}\omega_{ij}(\hat X_{ij} - X_{ij}^*)^2 + \mathsf{D}$ holds. Define
\begin{equation}
\label{eq:varsigma}
\varsigma= 96\pi^{-1}[\rho(\mathbb{E}\norm{\Sigma_R})^2+8].
\end{equation}
\begin{lemma}\label{lemma:rsc}
Let $\eta = 72\log(n+p)/(\pi\log(6/5))$ and $\rho>0$. With probability at least $1-8(n+p)^{-1}$, for all $X\in\mathcal{C}(\eta,\rho)$ we get
\begin{equation*}
\left|\Sigma(\omega, X)-\mathbb{E}\left[\Sigma(\omega, X)\right]\right|\leq 
\frac{\mathbb{E}\left[\Sigma(\omega, X)\right]}{2} + \varsigma,
\end{equation*}
with $\Sigma_R$ defined in \eqref{eq:def-SigmaR}.
\end{lemma}
\begin{proof}
Consider the event
\begin{equation*}
\mathcal{B}=\Bigg\{\sup_{X\in\mathcal{C}(\eta, \rho)}\left[\left|\Sigma(\omega, X)-\mathbb{E}\left[\Sigma(\omega, X)\right]\right|-
\frac{1}{2}\mathbb{E}\left[\Sigma(\omega, X)\right]\right] > \varsigma \Bigg\}.
\end{equation*}
Define also for $l\in\nset_*$
$$\mathcal{S}_l =  \left\{X\in \mathcal{C}(\eta, \rho); \kappa^{l-1}\eta < \mathbb{E}\left[\Sigma(\omega, X)\right] < \kappa^l\eta\right\},$$
for $\kappa=6/5$ and $\eta = 72\log(n+p)/(\pi\log(6/5))$. On $\mathcal{B}$, there exist $l\geq 1$ and $X\in\mathcal{C}(\eta,\rho)$ such that $X\in\mathcal{C}(\eta,\rho)\bigcap\mathcal{S}_l$, and
\begin{equation}
\label{eq:24}
\left|\Sigma(\omega, X)-\mathbb{E}\left[\Sigma(\omega, X)\right]\right|  > 
\frac{1}{2}\mathbb{E}\left[\Sigma(\omega, X)\right]+\varsigma >\frac{1}{2}\kappa^{l-1}\eta + \varsigma = \frac{5}{12}\kappa^l\eta + \varsigma.
\end{equation}
For $T>0$, define the set
$$\mathcal{C}(\eta, \rho, T) = \left\{X\in\mathcal{C}(\eta,\rho), \mathbb{E}\left[\Sigma(\omega, X)\right]\leq T\right\}$$
and the event 
$$\mathcal{B}_l = \left\{\sup_{X \in \mathcal{C}(\eta,\rho,\kappa^l\eta)}\left|\Sigma(\omega, X)-\mathbb{E}\left[\Sigma(\omega, X)\right]\right|>\frac{5}{12}\kappa^l\eta + \varsigma \right\}.$$
It follows from \eqref{eq:24} that $\mathcal{B}\subset \bigcup_{l=1}^{+\infty}\mathcal{B}_l$; thus, it is enough to estimate the probability of the events $\mathcal{B}_l$, $l\in\mathbb{N}$, and then apply the union bound. Such an estimation is given in the following Lemma, adapted from \citet{klopp:hal-01111757} (see Lemma 10). Define
\begin{equation}
\label{eq:def-Zt}
Z_T = \sup_{X\in \mathcal{C}(\eta, \rho, T)}\left|\Sigma(\omega, X)-\mathbb{E}\left[\Sigma(\omega, X)\right]\right|.
\end{equation}
\begin{lemma}
\label{lem:klopp2015}
Under the assumptions of \Cref{th:global-bound-1},
\begin{equation}
\label{eq:lemm:ZT}
\mathbb{P}\left( Z_T \geq \frac{5}{12}T + \varsigma\right)\leq 4\mathrm{e}^{-\pi T/72},
\end{equation}
where $\varsigma$ is defined in \eqref{eq:varsigma}.
\end{lemma}
\begin{proof} We use the following Talagrand's concentration inequality and a symmetrization argument. Recall the statement of Talagrand's concentration inequality. Let $f:[-1,1]^{m}\mapsto \mathbb{R}$ a convex Lipschitz function with Lipschitz constant L, $\Xi_1,\ldots,\Xi_m$ be independent random variables taking values in $[-1,1]$, and $Z:=f(\Xi_1,\ldots,\Xi_m)$. Then, for any $t\geq 0$, $\mathbb{P}(|Z-\mathbb{E}[Z]|\geq 16L + t)\leq 4\mathrm{e}^{-t^2/2L^2}$. For $\mathbf{x} = (x_{ij})$, $(i,j)\in\intervi{n}\times\intervi{p}$, we apply this result to the function $$f(\mathbf{x}) = \sup_{X\in \mathcal{C}(\eta, \rho, T)}\left|\sum_{(i,j)\in\intervi{n}\times\intervi{p}}(x_{ij}-\pi_{ij})X_{ij}^2 \right|,$$ which is Lipschitz with Lipschitz constant $\sqrt{\pi^{-1}T}$:
\begin{equation*}
\begin{aligned}
& \left|f(x_{11},\ldots,x_{np})-f(z_{11},\ldots,z_{np})\right|\\
& \quad = \left|\sup_{X\in \mathcal{C}(\eta, \rho, T)}\left|\sum_{(i,j)\in\intervi{n}\times\intervi{p}}(x_{ij}-\pi_{ij})X_{ij}^2\right|-\sup_{X\in \mathcal{C}(\eta, \rho, T)}\left|\sum_{(i,j)\in\intervi{n}\times\intervi{p}}(z_{ij}-\pi_{ij})X_{ij}^2\right| \right|\\
& \quad \leq \sup_{X\in \mathcal{C}(\eta, \rho, T)}\left|\left|\sum_{(i,j)\in\intervi{n}\times\intervi{p}}(x_{ij}-\pi_{ij})X_{ij}^2\right|-\left|\sum_{(i,j)\in\intervi{n}\times\intervi{p}}(z_{ij}-\pi_{ij})X_{ij}^2\right| \right|\\
&\quad \leq \sup_{X\in \mathcal{C}(\eta, \rho, T)}\left|\sum_{(i,j)\in\intervi{n}\times\intervi{p}}(x_{ij}-\pi_{ij})X_{ij}^2-\sum_{(i,j)\in\intervi{n}\times\intervi{p}}(z_{ij}-\pi_{ij})X_{ij}^2 \right|\\
&\quad \leq\sup_{X\in \mathcal{C}(\eta, \rho, T)}\left|\sum_{(i,j)\in\intervi{n}\times\intervi{p}}(x_{ij}-z_{ij})X_{ij}^2 \right|\\
&\quad \leq\sup_{X\in \mathcal{C}(\eta, \rho, T)}\sqrt{(i,j)\in\sum_{\intervi{n}\times\intervi{p}}\pi_{ij}^{-1}(x_{ij}-z_{ij})^2}\sqrt{\sum_{(i,j)\in\intervi{n}\times\intervi{p}}\pi_{ij}X_{ij}^4}\\
&\quad \leq\sup_{X\in \mathcal{C}(\eta, \rho, T)}\sqrt{\pi^{-1}}\sqrt{\sum_{(i,j)\in\intervi{n}\times\intervi{p}}(x_{ij}-z_{ij})^2}\sqrt{\sum_{(i,j)\in\intervi{n}\times\intervi{p}}\pi_{ij}X_{ij}^2}\\
& \quad \leq \sqrt{\pi^{-1}T}\sqrt{\sum_{(i,j)\in\intervi{n}\times\intervi{p}}(x_{ij}-z_{ij})^2},
\end{aligned}\\
\end{equation*}
where we have used $||a|-|b||\leq |a-b|$,$\norm{X}[\infty]\leq 1$ and $\mathbb{E}\left[\Sigma(\omega, X)\right]\leq T$. 
Thus, Talagrand's inequality and the identity $\sqrt{\pi^{-1}T}\leq T/(2\times 96)+96/(2\pi)$ give
$$\mathbb{P}\left(Z_T\geq \mathbb{E}(Z_T)+768\pi^{-1}+\frac{1}{12}T+t \right)\leq 4\mathrm{e}^{-t^2\pi/2T}.$$
Taking $t = T/6$ we get 
\begin{equation}
\label{eq:talagr}
\mathbb{P}\left(Z_T\geq \mathbb{E}(Z_T)+768\pi^{-1}+\frac{3}{12}T\right)\leq 4\mathrm{e}^{-\pi T/72}.
\end{equation}
Now we bound the expectation $\mathbb{E}[Z_T]$ using a symmetrization argument \citep[Section 7.2]{Ledoux2001}. Let~$(\epsilon_{ij})$ be an i.i.d.~Rademacher sequence. We have
\begin{equation}
\mathbb{E}(Z_T) \leq 2\mathbb{E}\left(\sup_{X\in\mathcal{C}(\eta,\rho,T)}\left|\sum_{(i,j)\in\intervi{n}\times\intervi{p}}\epsilon_{ij}\omega_{ij}X_{ij}^2\right| \right),
\end{equation}
Then, the contraction inequality (see \citet{Koltchinskii2011}, Theorem 2.2) yields 
$$\mathbb{E}(Z_T)\leq 8\mathbb{E}\left(\sup_{X\in\mathcal{C}(\eta,\rho,T)}\left|\sum_{(i,j)\in\intervi{n}\times\intervi{p}}\epsilon_{ij}\omega_{ij}X_{ij} \right| \right) =  8\mathbb{E}\left(\sup_{X\in\mathcal{C}(\eta,\rho,T)}\left|\left\langle \Sigma_R, X \right\rangle \right| \right),$$
where $\Sigma_R $ is defined in \eqref{eq:def-SigmaR}. For $X\in \mathcal{C}(\eta,\rho,T)$ we have that $\norm{X}[*]\leq \sqrt{\rho\pi^{-1}T}$. Then by duality between the nuclear and operator norms we obtain
$$ \mathbb{E}(Z_T)\leq 8\mathbb{E}\left(\sup_{\norm{X}[*]\leq \sqrt{\rho\pi^{-1}T}}\left|\left\langle \Sigma_R, X \right\rangle \right| \right) \leq 8\sqrt{\rho\pi^{-1}T}\mathbb{E}\norm{\Sigma_R}.$$
Combined with \eqref{eq:talagr} and using $8\sqrt{\rho\pi^{-1}T}\mathbb{E}\norm{\Sigma_R} \leq \frac{T}{2\times 3}+\frac{3\times 8^2\rho \pi^{-1}}{2}(\mathbb{E}\norm{\Sigma_R})^2$ we finally obtain
\eqref{eq:lemm:ZT} using the definition of $\varsigma$ in \eqref{eq:varsigma}.
\end{proof}
\Cref{lem:klopp2015} implies that $$\mathbb{P}(\mathcal{B})\leq \sum_{l=1}^{+\infty}\mathbb{P}(\mathcal{B}_l)\leq 4\sum_{l=1}^{+\infty} \exp(-\pi\kappa^l\eta/72)\leq 8/(n+p),$$
which concludes the proof.
\end{proof}

\paragraph{Case 1} If $\sum_{(i,j)\in\intervi{n}\times\intervi{p}} \pi_{ij}(\hat X_{ij}-X^*_{ij})^2\leq \eta$, then $\norm{\hat X-X^*}[2]^2\leq \eta/\pi$ and the result of \Cref{th:global-bound-1} \eqref{eq:global-bound} is proved.

\paragraph{Case 2} If $\sum_{(i,j)\in\intervi{n}\times\intervi{p}} \pi_{ij}(\hat X_{ij}-X^*_{ij})^2>\eta$. Let us show that $(\hat X - X^*)/2\gamma\in\mathcal{C}(\eta, 64\operatorname{rank}(X^*))$. Using \eqref{eq:one}, $\sigma_{-}^2\sum_{(i,j)\in\Omega}(\hat X_{ij}-X^*_{ij})^2\geq 0$ and $\norm{\nabla\mathcal{L}(X^*)}\leq \lambda/2$, we obtain that
\begin{equation*}
\norm{\mathcal{P}_{\Theta^*}^{\perp}(\hat\Theta - \Theta^*)}[*] \leq 3\norm{\mathcal{P}_{\Theta^*}(\hat\Theta - \Theta^*)}[*] + \norm{\hat X_0-X_0^*}[*].
\end{equation*}
On the other hand,
\begin{equation*}
\begin{aligned}
&\norm{\hat X-X^*}[*]&&\leq \norm{\mathcal{P}_{\Theta^*}^{\perp}(\hat\Theta - \Theta^*)}[*] +\norm{\mathcal{P}_{\Theta^*}(\hat\Theta - \Theta^*)}[*]+\norm{\hat X_0-X_0^*}[*]\\
& && \leq 4\norm{\mathcal{P}_{\Theta^*}(\hat\Theta - \Theta^*)}[*]+2\norm{\hat X_0-X_0^*}[*]\\
& && \leq 2\sqrt{2\operatorname{rank}(\Theta^*)}\norm{\hat\Theta - \Theta^*}[F] + 2\sqrt{r}\norm{\hat X_0- X_0^*}[F]\\
& && \leq \sqrt{64\operatorname{rank}(X^*)}\norm{\hat X- X^*}[F].
\end{aligned}
\end{equation*}
Thus, \Cref{lemma:rsc} implies that with probability at least $1-8(n+p)^{-1}$, 
\begin{equation}
\label{eq:three}
\sum_{(i,j)\in\intervi{n}\times\intervi{p}}\omega_{ij}(\hat X_{ij} - X^*_{ij})^2\geq \frac{\mathbb{E}[\sum_{(i,j)\in\intervi{n}\times\intervi{p}}\omega_{ij}(\hat X_{ij} - X^*_{ij})^2]}{2}-384\gamma^2\pi^{-1}[64\operatorname{rank}(X^*)(\mathbb{E}\norm{\Sigma_R})^2+8].
\end{equation}
Combining \eqref{eq:three} and \eqref{eq:lambda-bound} we obtain
\begin{equation*}
\frac{\pi\norm{\hat X-X^*}[F]^2}{2} - \frac{384\gamma^2[64\operatorname{rank}(X^*)(\mathbb{E}\norm{\Sigma_R})^2+8]}{\pi} \leq \frac{\lambda}{\sigma_-^2}\left(3\sqrt{2\operatorname{rank}(\Theta^*)}+\sqrt{r}\right)\norm{\hat X-X^*}[F].
\end{equation*}
Finally, using the identity $ab\leq a^2+b^2/4$ and $\operatorname{rank}(X^*)\leq \operatorname{rank}(\Theta^*)+r $ we obtain
\begin{equation}
\label{eq:res1}
\norm{\hat X-X^*}[F]^2 \leq \left(\frac{192\lambda^2}{\pi^2\sigma_{-}^4} + \frac{24576\gamma^2(\mathbb{E}\norm{\Sigma_R})^2}{\pi^2}\right) [\mathrm{rk}(\Theta^*)+r] + \frac{6144}{\pi^2}.
\end{equation}

\subsection{Proof of \Cref{th:global-bound}}
\label{proof:global-bound}
\Cref{th:global-bound} derives from \Cref{th:global-bound-1} and combining the two following steps: 1) computing a value of $\lambda$ such that the condition $\lambda\geq 2\norm{\nabla \mathcal{L}(X^*)}$ holds with high probability and 2) controlling $\mathbb{E}\norm{\Sigma_R}$. Let us start with 1). Define the random matrices $Z_{ij}=\omega_{ij}(-Y_{ij}+\exp(X^*_{ij}))E_{ij}$ and the quantity
\begin{equation}
\label{eq:sigmaZ}
\sigma_Z^2=\max\left(\frac{1}{np}\left\|\sum_{i=1}^n\sum_{j=1}^p\mathbb{E}\left[Z_{ij}Z_{ij}^{\top}\right] \right\|,\quad \frac{1}{np}\left\|\sum_{i=1}^n\sum_{j=1}^p\mathbb{E}\left[Z_{ij}^{\top}Z_{ij}\right] \right\|\right).
\end{equation}
\begin{lemma}
\label{lem:sigmaZ}
Under the assumptions of \Cref{th:global-bound},
\begin{equation}
\frac{\sigma_-^2\beta}{np}\leq \sigma_Z^2\leq \frac{\sigma_+^2\beta}{np}.
\end{equation}
\end{lemma}
\begin{proof}
For all $(i,j)\in\intervi{n}\times\intervi{p}$, $Z_{ij}Z_{ij}^T = \omega_{ij}(-Y_{ij}+\exp(X_{ij}^*))^2E_{ij}E_{ij}^{\top}$, and $\mathbb{E}[Z_{ij}Z_{ij}^{\top}] = \mathbb{E}[\omega_{ij}]\mathbb{E}[(-Y_{ij}+\exp(X_{ij}^*))^2]E_{ij}E_{ij}^{\top}$, which is a diagonal matrix with $0$ everywhere except on the $i$-th element of its diagonal, where its value is $\mathbb{E}[\omega_{ij}]\mathbb{E}[(-Y_{ij}+\exp(X_{ij}^*))^2]$. Thus,
$$\sum_{(i,j)\in\intervi{n}\times\intervi{p}} \mathbb{E}[Z_{ij}Z_{ij}^{\top}]$$
is also a diagonal matrix, and the $i$-th element of its diagonal is $\sum_{j=1}^p\mathbb{E}[\omega_{ij}]\mathbb{E}[(-Y_{ij}+\exp(X_{ij}^*))^2]$. We obtain that
\begin{equation*}
\frac{1}{np}\left\|\sum_{(i,j)\in\intervi{n}\times\intervi{p}} \mathbb{E}[Z_{ij}Z_{ij}^{\top}] \right\| = \frac{1}{np}\max_{i\in\intervi{n}}\sum_{j=1}^p\mathbb{E}[\omega_{ij}]\mathbb{E}[(-Y_{ij}+\exp(X_{ij}^*))^2].
\end{equation*}
Using $\mathbb{E}[Y_{ij}] =\exp(X_{ij}^*) $ and $\sigma_-^2\leq\text{var}(Y_{ij})\leq \sigma_+^2$, we obtain:
\begin{equation}
\label{eq:sigmaZ-1}
\frac{\sigma_-^2}{np}\max_{i\in\intervi{n}}\sum_{j=1}^p\mathbb{E}[\omega_{ij}]\leq \frac{1}{np}\left\|\sum_{(i,j)\in\intervi{n}\times\intervi{p}} \mathbb{E}[Z_{ij}Z_{ij}^{\top}] \right\| \leq \frac{\sigma_+^2}{np}\max_{i\in\intervi{n}}\sum_{j=1}^p\mathbb{E}[\omega_{ij}].
\end{equation}
Using the same arguments, we also obtain
\begin{equation}
\label{eq:sigmaZ-2}
\frac{\sigma_-^2}{np}\max_{j\in\intervi{p}}\sum_{i=1}^n\mathbb{E}[\omega_{ij}]\leq \frac{1}{np}\left\|\sum_{(i,j)\in\intervi{n}\times\intervi{p}} \mathbb{E}[Z_{ij}^{\top}Z_{ij}] \right\| \leq \frac{\sigma_+^2}{np}\max_{j\in\intervi{p}}\sum_{i=1}^n\mathbb{E}[\omega_{ij}].
\end{equation}
Combining \eqref{eq:sigmaZ-1} and \eqref{eq:sigmaZ-1}, we obtain that 
\begin{multline*}
\frac{\sigma_-^2}{np}\max\left\{\max_{i\in\intervi{n}}\sum_{j=1}^p\mathbb{E}[\omega_{ij}], \max_{j\in\intervi{p}}\sum_{i=1}^n\mathbb{E}[\omega_{ij}]\right\}\leq  \sigma_Z^2\leq\\
 \frac{\sigma_+^2}{np}\max\left\{\max_{i\in\intervi{n}}\sum_{j=1}^p\mathbb{E}[\omega_{ij}], \max_{j\in\intervi{p}}\sum_{i=1}^n\mathbb{E}[\omega_{ij}]\right\},
\end{multline*}
which concludes the proof.
\end{proof}
Note that
$\mathbb{E}\left[Z_{ij}\right]=0$ for all $(i,j)\in\intervi{n}\times\intervi{p}$ and $ \nabla\mathcal{L}(X^*)=\sum_{i=1}^n\sum_{j=1}^pZ_{ij}.$ We use an extension of Theorem 4 in \citet{koltchinskii} to rectangular matrices via self-adjoint dilation (cf., for example, 2.6 in \citet{Tropp2012}). Let $\Xi_1,\ldots,\Xi_m$ be $m$ independent $(n\times p)$-matrices satisfying $\mathbb{E}[\Xi_i] = 0$ and
$$\inf \{K>0:\mathbb{E}[\exp(\norm{\Xi_i}/K)]\leq e\}<M$$
for some constant $M$ and for all $i\in\{1,\ldots,m\}$. Define
$$\sigma^2=\max\left(\frac{1}{m}\left\|\sum_{i=1}^m\mathbb{E}\left(\Xi_{i}\Xi_{i}^T\right) \right\|,\quad \frac{1}{m}\left\|\sum_{i=1}^m\mathbb{E}\left(\Xi_{i}^T\Xi_{i}\right) \right\|\right),$$
and $\bar U = M\log(1+2 \frac{M^2}{\sigma^2})$. Then, for $t \bar U \leq 2(e-1)\sigma^2m$,
$$ \mathbb{P}\left\{\left\|\frac{1}{m}\sum_{i=1}^m\Xi_i\right\| \geq t \right\}\leq 2(n+p)\exp\left\{  -\frac{t^2}{4m\sigma^2+2\bar U t/3}\right\}$$
and for $t\bar U > 2(e-1)\sigma^2m$,
$$ \mathbb{P}\left\{\left\|\frac{1}{m}\sum_{i=1}^m\Xi_i\right\| \geq t \right\}\leq 2(n+p)\exp\left\{  -\frac{t}{(e-1)\bar U }\right\}.$$
Under Assumption~\ref{ass:orlicz} we may apply this result with $m=np$, $(\Xi_1,\ldots,\Xi_m) = (Z_{11},\ldots,Z_{np})$, $M = 2\delta$, $\sigma^2 =\sigma_Z^2$ and $\bar U = 2\delta \log(1+8\delta^2/\sigma_Z^2)$. Taking
$$t\geq \max\left\{2\sigma_Z\sqrt{3np\log(n+p)}, 6\delta(e-1)\log(1+8\delta^2/\sigma_Z^2)\log(n+p)\right\}$$
and using \Cref{lem:sigmaZ}, we get that with probability at least $1-(n+p)^{-1}$,
\begin{equation*}
\norm{\nabla\mathcal{L}(X^*)}\leq \max\left\lbrace 2\sigma_{+}\left(3\beta\log(n+p)\right)^{1/2}, 6\delta(e-1)\log\{1+8\delta^2np/(\beta\sigma_-^2)\}\log(n+p) \right\rbrace.
\end{equation*}
Thus, taking $\lambda$ as in \Cref{th:global-bound} ensures that $\lambda\geq 2\norm{\nabla \mathcal{L}(X^*)}$ with probability at least $1-(n+p)^{-1}$.\\

We now control $\mathbb{E}\norm{\Sigma_R}$ with the following lemma.
\begin{lemma}
\label{lemma:SigmaR}
There exists an absolute constant $C^*$ such that the two following inequality holds
\begin{equation*}
\label{eq:sigmaR-infty}
\mathbb{E}[\norm{\Sigma_R}] \leq C^*\left\{\sqrt{\beta} + \sqrt{\log m}\right\}.
\end{equation*}
\end{lemma}
\begin{proof}
We use an extension to rectangular matrices via self-adjoint dilation of Corollary 3.3 in \citet{bandeira2016}.
\begin{proposition}
\label{prop:sigmaR-nuc}
Let $A$ be an $n\times p$ rectangular matrix with $A_{ij}$ independent centered bounded random variables. then, there exists a universal constant $C^*$ such that
$$\mathbb{E}[{\norm{A}}] \leq C^*\left\{\sigma_1\vee \sigma_2 + \sigma_*\sqrt{\log(n\wedge p)}\right\},$$
\begin{equation*}
\sigma_1 = \max_i\sqrt{\sum_j\mathbb{E}[{A_{ij}^2}]},\quad
\sigma_2 = \max_j\sqrt{\sum_i\mathbb{E}[{A_{ij}^2}]},\quad
\sigma_* = \max_{i,j}|A_{ij}|.
\end{equation*}
\end{proposition}
Applying \Cref{prop:sigmaR-nuc} to $\Sigma_R$ with $\sigma_1 \vee \sigma_2 \leq \sqrt{\beta}/|\Omega|$ and $\sigma_*\leq 1$ we obtain
$$\mathbb{E}[{\norm{\Sigma_R}}] \leq C^*\left\{\sqrt{\beta} + \sqrt{\log(n\wedge p)}\right\}.$$
\end{proof}
Combining 1) and 2) with \eqref{eq:res1} and a union bound argument, we obtain the result of \Cref{th:global-bound}.

\subsection{Proof of \Cref{prop:null-lambda}}
\label{lambda-proof}
In what follows we denote for $X_0\in\mathcal{X}_0$ and $\Theta\in\mathcal{T}$
$\mathcal{F}^{\lambda}(X_0,\Theta) = \mathcal{L}(X_0+\Theta) + \lambda\norm{\Theta}[*].$
We establish below that $\lambda_0(Y)$ defined in \eqref{eq:null-lambda} is equal to
$$\lambda_0(Y) = \underset{\lambda}{\min} \quad 0\in \partial_{\Theta}\lbrace\mathcal{F}^{\lambda}(\hat X_0, \Theta) + \chi_{\mathcal{T}}(\Theta)\rbrace\mid_{\Theta=0},$$
where for  $\mathcal{K}\subset\mathbb{R}^{n\times p}$ ,$\chi_{\mathcal{K}}(X)$ is the characteristic function of the set $\mathcal{K}$, equal to $0$ on $\mathcal{K}$ and $+\infty$ elsewhere,
and $\hat{X}_0=\underset{X \in \mathcal{X}_0}{\operatorname{arg min}} \mathcal{L}(X)$ (see \eqref{eq:null-lambda}). The subdifferential of the objective function $\mathcal{F}^{\lambda}$ with respect to $\Theta$ is given by
$$\partial_{\Theta}\mathcal{F}^{\lambda}(\hat X_0,0) =\nabla\mathcal{L}(\hat X_0+\Theta)\mid_{\Theta=0}+\lambda \partial_{\Theta}\left\|\Theta \right\|_{*}\mid_{\Theta=0} +\partial_{\Theta}\chi_{\mathcal{T}}(\Theta)\mid_{\Theta=0}.$$
 $0\in\partial_{\Theta}\chi_{\mathcal{T}}(\Theta)\mid_{\Theta=0}$. \Cref{lemma:nuc-subdiff} ensures that $0\in \partial\mathcal{F}^{\lambda}(\Theta)\mid_{\Theta=0}$ if and only if $$0 \in \left\{\nabla\mathcal{L}(\hat X_0)+ \lambda W\text{; } \left\|\mathcal{P}_{\mathcal{T}}(W)\right\|\leq 1\right\}.$$ This is equivalent to $\lambda \geq \left\|\mathcal{P}_{\mathcal{T}}(\nabla\mathcal{L}(\hat{X}_0))\right\|.$ Additionally, at the optimum $\hat{X}_0$, we have $\mathcal{P}_{\mathcal{T}}(\nabla\mathcal{L}(\hat{X}_0)) = \nabla\mathcal{L}(\hat{X}_0)$, which concludes the proof.\\

\begin{lemma}
\label{lemma:nuc-subdiff}
Let $g:\mathcal{T} \rightarrow \mathbb{R}_+$ be the function defined by $g(A)=\left\|A\right\|_{*}$ for $A\in \mathcal{T}$.
$\partial g(0)= \left\{W\in \mathbb{R}^{n\times p}, \left\|\mathcal{P}_{\mathcal{T}}(W)\right\|\leq 1 \right\}$.
\end{lemma}
\begin{proof}
By definition of the subdifferential we need to prove that for all $W\in \mathbb{R}^{n\times p}$, $\left\|\mathcal{P}_{\mathcal{T}}(W)\right\|<1$, and for all $B\in \mathcal{T}$, $g(B)\geq g(0) + \left\langle W, B-0  \right\rangle$. First $B\in \mathcal{T}$ implies $\langle W,B\rangle = \langle \mathcal{P}_{\mathcal{T}}(W),B\rangle$, therefore $\left\|\mathcal{P}_{\mathcal{T}}(W)\right\|\leq 1$ is a sufficient condition for $W\in \partial g(0)$. Now assume $\left\|\mathcal{P}_{\mathcal{T}}(W) \right\|>1$ and let $\mathcal{P}_{\mathcal{T}}(W)=U\Sigma V^T$, where $U$ and $V$ are orthogonal matrices of left and right singular vectors, and $\Sigma_{11} = \left\|\mathcal{P}_{\mathcal{T}}(W) \right\|>1$. Let us define $B=U\tilde{\Sigma}V^T$, $\tilde{\Sigma}_{11}=1$ and $\tilde{\Sigma}_{ij}=0$ elsewhere; note that with this definition $B\in\mathcal{T}$. We have $g(B)=1$ and $\langle \mathcal{P}_{\mathcal{T}}(W), B\rangle = \Sigma_{11}>g(B)$. Therefore $\left\|\mathcal{P}_{\mathcal{T}}(W) \right\| > 1\Rightarrow W\notin \partial g(0)$, from which we conclude
$$\partial g(0)= \left\{W\in \mathbb{R}^{n\times p}, \left\|\mathcal{P}_{\mathcal{T}}(W)\right\|<1 \right\}.$$
\end{proof}



\bibliographystyle{plainnat}
\bibliography{references}

\end{document}